\documentclass{article}
\usepackage[utf8]{inputenc}
\usepackage{amsthm}
\usepackage{amsmath}
\usepackage{xcolor}
\usepackage{xspace}
\usepackage{hyperref}
\usepackage{mathtools}

\usepackage[subtle, mathspacing=normal]{savetrees} 

\title{Iceberg Hashing: Optimizing Many Hash-Table Criteria at Once}

\author{
\begin{tabular}{@{}c@{}}
Michael A. Bender\thanks{Stony Brook University. \href{mailto:bender@cs.stonybrook.edu}{\texttt{bender@cs.stonybrook.edu}}} \qquad Alex Conway\thanks{Cornell Tech. \href{mailto:ajc473@cornell.edu}{\texttt{ajc473@cornell.edu}}} \qquad Mart\'{\i}n Farach-Colton \thanks{Rutgers University. \href{mailto:martin@farach-colton.com}{\texttt{martin@farach-colton.com}}}\vspace{0.22cm}\\
William Kuszmaul \thanks{Massachusetts Institute of Technology. \href{mailto:kuszmaul@mit.edu}{\texttt{kuszmaul@mit.edu}}} \qquad Guido Tagliavini \thanks{Snowflake Inc. \href{mailto:guido.tagliavini@snowflake.com}{\texttt{guido.tagliavini@snowflake.com}}}
\end{tabular}
}

\date{}

\usepackage[letterpaper, margin=1in]{geometry}
\usepackage{xspace}
\usepackage{enumitem}
\usepackage{cleveref}
\usepackage{tcolorbox}
\usepackage{amsmath,amsthm,amsfonts}
\tcbuselibrary{skins,breakable}

\usepackage{algpseudocode,algorithm,algorithmicx}
\newcommand*\Let[2]{\State #1 $\gets$ #2}
\algrenewcommand\algorithmicrequire{\textbf{Description:}}
\algrenewcommand\algorithmicensure{\textbf{Postcondition:}}
\algnewcommand{\IIf}[1]{\State\algorithmicif\ #1\ \algorithmicthen}
\algnewcommand{\EndIIf}{\unskip\ \algorithmicend\ \algorithmicif}

\usepackage{listings}
\usepackage{thmtools}
\usepackage{thm-restate}
\usepackage{lipsum}
\usepackage[normalem]{ulem}

\DeclareMathOperator{\E}{\mathbb{E}}

\renewcommand{\epsilon}{\varepsilon}
\usepackage{todonotes}

\newcommand{\defn}[1]{\textbf{\emph{#1}}}
\newcommand{\poly}{\operatorname{poly}}
\newcommand{\polylog}{\operatorname{polylog}}
\newcommand{\prob}[1]{\Pr\left[#1\right]}

\newcommand{\expect}[1]{\E\left[#1\right]}

\newcommand{\eps}{\varepsilon}
\newcommand{\paren}[1]{\left(#1\right)}

\newcommand{\bin}{\text{Bin}}

\newcommand{\g}{\textsf{bin}}
\newcommand{\f}{\textsf{fp}}
\newcommand{\rr}{r}

\newcommand{\historyless}{\mbox{\textsc{Iceberg Game}}\xspace}

\newcommand\numberthis{\addtocounter{equation}{1}\tag{\theequation}}

\newtheorem{thm}{Theorem}
\newtheorem{lem}{Lemma}
\newtheorem{cor}{Corollary}

\newtheorem{clm}{Claim}
\newtheorem{rmk}{Remark}

\theoremstyle{definition}

\newcommand{\punt}[1]{}

\newcommand{\calF}{\mathcal{F}}

\newcommand{\calT}{\mathcal{T}}

\newcommand{\word}{w}

\renewcommand{\paragraph}[1]{\medskip\smallskip\noindent\textbf{#1}}

\begin{document}
\maketitle
\thispagestyle{empty}
\begin{abstract}

Despite being one of the oldest data structures in computer science, hash tables continue to be the focus of a great deal of both theoretical and empirical research. A central reason for this is that many of the fundamental properties that one desires from a hash table are difficult to achieve simultaneously; thus many variants offering different trade-offs have been proposed.

This paper introduces Iceberg hashing, a hash table that simultaneously offers the strongest known guarantees on a large number of core properties. Iceberg hashing supports constant-time operations while improving on the state of the art for space efficiency, cache efficiency, and low failure probability.  Iceberg hashing is also the first hash table to support a load factor of up to $1 - o(1)$ while being stable, meaning that the position where an element is stored only ever changes when resizes occur. In fact, in the setting where keys are $\Theta(\log n)$ bits, the space guarantees that Iceberg hashing offers, namely that it uses at most $\log \binom{|U|}{n} + O(n \log \log n)$ bits to store $n$ items from a universe $U$, matches a lower bound by Demaine et al. that applies to any stable hash table.

Iceberg hashing introduces new general-purpose techniques for some of the most basic aspects of hash-table design.  Notably, our indirection-free technique for dynamic resizing, which we call waterfall addressing, and our techniques for achieving stability and very-high probability guarantees, can be applied to any hash table that makes use of the front-yard/backyard paradigm for hash table design.

\end{abstract}

\newpage 
\pagenumbering{arabic}

\section{Introduction}
\label{sec:intro}

The hash table is one of the oldest and most fundamental data structures in computer science.
Hash tables were invented by Hans Peter Luhn in 1953 during the development of  IBM's first commercial
scientific computer, the IBM 701 \cite{KnuthVol3}. Luhn's implementation used what is now known as
chained hashing.\footnote{As Knuth points out in \cite{KnuthVol1}, the implementation may have also been the
first use of linked lists in computer science.}
Since items are addressed via pointers, chained hash tables waste space and offer poor data locality.

In the nearly seven decades since, there has been a huge literature on hashing; some important milestones can be summarized in the following progression of work. 
Linear probing, which was introduced in 1954~\cite{KnuthVol3, peterson57}, achieves good data locality and constant-time operations in expectation, but scales poorly to high load factors.  In the 1980s, Fredman, Koml\'os, and Szemer\'edi~\cite{Fredman82FKS} showed how to achieve worst-case constant time queries and subsequent work~\cite{dietzfelbinger1990new, dietzfelbinger88hash} showed how to dynamize this hash table, but at the cost of poor space efficiency and data locality.  In the early 2000s, Cuckoo hashing~\cite{Pagh:CuckooHash, Fotakis03Cuckoo, DietzfelbingerWe07, Arbitman09Deamortized}  was introduced, providing constant-time queries and updates with better space efficiency.  Finally, in the past decade, several hash tables 
\cite{arbitman2010backyard, liu2020succinct, goodrich2011fully, goodrich2012cache} have been developed that offer a variety of even stronger performance guarantees, including very-high probability constant-time operations, very high load factors, etc.

One of the great ironies in the study of hashing is that,
even after seven decades of research and many proposed alternatives, chained hashing remains one of the most widely used hash table designs, even serving as the default in performance-oriented languages such as C++ \cite{cplusplus1, cplusplus2}.
Since chaining is missing many of the desirable properties of other hash tables (space efficiency, data locality, constant-time operations, etc.), why is it that it continues to be so widely used?

Chaining offers one guarantee that is not offered by other hash table designs: referential stability.
{Referential stability} requires that elements not change location in the table,  except for when table resizes are performed \cite{sandersstability, cplusplus1, cplusplus2, KnuthVol3, originalstability}.
This is important in many settings: to reduce locking and increase concurrency; to allow pointers into the table; to support iterators through the hash table, etc.

What makes stability algorithmically interesting is that the known techniques for achieving it are fundamentally at odds with
the other desirable guarantees. Stability itself is easily achieved by storing pointers to elements in the hash table, rather than the elements themselves.  
But as in chaining, these pointers compromise other central guarantees, such as space efficiency and data locality. 

In fact, stability illustrates just one example of a more general phenomenon---that known techniques for achieving many central hash-table guarantees preclude others. 
Even in cases where we know how to achieve individual guarantees, the question of whether we can get these guarantees together in the same hash table is often much harder. Some of the most substantial breakthroughs in the field have been needed to achieve even basic combinations, e.g., high load factor and dynamic resizing \cite{Raman03Succinct}, high load factor and constant-time operations \cite{arbitman2010backyard}, or very recently, dynamically-resizable high load factor and constant-time operations \cite{liu2020succinct}. And, as we shall discuss in more detail later, some other basic combinations are still well beyond the known techniques. 

Modern work on hashing~\cite{jensen2008optimality, arbitman2010backyard, liu2020succinct, goodrich2011fully, goodrich2012cache, sandersstability} focuses on the following core list of desirable guarantees:

\medskip
\noindent{\textbf{Time:}}
\begin{itemize}[noitemsep,nolistsep]
\item \textbf{Constant-time operations:} insertions/queries/deletions take $O(1)$ time w.h.p.
\item \textbf{\boldmath $(1 + o(1))$ cache optimality:} operations incur $1 + o(1)$ cache misses in the external memory model.  
\end{itemize}
\medskip
\noindent{\textbf{Space:}}
\begin{itemize}[noitemsep,nolistsep]
\item \textbf{Load factors of \boldmath $1 - o(1)$:} all but a $o(1)$ fraction of space is used to store elements. 
\item \textbf{Dynamic resizing:} the table dynamically adjusts its space consumption to match the current size. 
\end{itemize}
\medskip
\noindent{\textbf{Functionality:}}
\begin{itemize}[noitemsep,nolistsep]
\item \textbf{Very-high probability guarantees:} the guarantees have subpolynomial failure probability.
\item \textbf{Referential stability:} the only way that elements move around is when the table is resized.  
\end{itemize}
\medskip
\medskip

Each property individually has its own (sometimes extensive) line of research, 
and the question of whether optimal guarantees for all of the properties can be achieved together has remained a significant open problem.

\paragraph{This paper: Iceberg hashing.}
In this work we introduce \defn{Iceberg hash tables}. Iceberg hashing matches the states of the art for all of the above properties
\emph{simultaneously}, and also improves the states of the art for space efficiency and failure probability. 

Iceberg hashing introduces new techniques for some of the most basic aspects of hash-table design. 
Notably, our indirection-free technique for dynamic resizing, which we call \defn{waterfall addressing}, and our techniques for achieving stability and very-high probability 
guarantees, can be applied to any hash table that makes use of the backyarding paradigm for hash table design.

Iceberg hashing also revisits one of the oldest approaches for designing space-efficient hash tables: backyarding. 
Introduced in the 1950s~\cite{peterson57, vitter1987design}, the basic idea is that records are first hashed into bins in the \defn{front yard} and if the target bin is full, the record is instead stored in a
small \defn{backyard} hash table. As long as the backyard is small, consisting of $o(n)$ elements, we can afford to store it
in a less space-efficient manner. In recent work, backyarding has been used to achieve high space efficiency in constant-time hash tables 
\cite{arbitman2010backyard, BenderFaGo18, goodrich2011fully, goodrich2012cache}. Our techniques allow for this space efficiency to be preserved, while also achieving the other core guarantees described above. 

\subsection{The guarantees of an Iceberg hash table}

\paragraph{Referential stability.}
A hash table is said to be \defn{stable} if whenever a new element  $x$ is inserted, the position
in which $x$ (along with any value associated with $x$) is stored is guaranteed not to change until either $x$ is deleted or the table
is resized~\cite{sandersstability,originalstability}.\footnote{In addition to being required for any implementation of the C++ unordered map~\cite{cplusplus1}, stability is an integral part of the design for the standard hash tables used at both Google \cite{abseil}
and Facebook \cite{F14}. Stable hash tables typically offer a \texttt{Reserve} function, which allows users to guarantee that the table
will remain stable until the next time that it exceeds some reserved capacity, which is why stability is typically not required during resizing.}

Empirical work on the problem of designing space-efficient, stable hash tables dates back to the early 1980s~\cite{originalstability, sandersstability, abseil, F14} (see also Knuth's Volume 3 \cite{KnuthVol3}). Much of the theoretical work on stability has focused on a weaker version of the property called \defn{value stability}: values associated with keys are stable, but the keys need not be stored with those values and are allowed to move.\footnote{Value stability is sufficient for some applications of stability (e.g., storing pointers to values, so that the values can be directly edited) but not others (e.g., supporting iterators, storing pointers into the hash table that can be used to verify that a given key/value is present; designing concurrent hash tables that rely on elements staying put, etc.).} Demaine et al.~give a general-purpose approach (Theorem 3 of \cite{demaine2005dynamic}) for space-efficiently achieving value stability in any hash table, by adding an extra layer of indirection that can be encoded with just $O(\log \log n)$ extra bits per key. Of course, such a layer of indirection is incompatible with data-locality, so if we want to achieve value stability (and, more generally, full stability) in a hash table that is also cache friendly, then an alternative approach must be taken.

Besides the approach of using indirection \cite{cplusplus1, cplusplus2}, a second common approach to achieving stability has been to consider open addressing schemes (such as linear probing) 
with deletions implemented using \defn{tombstones}; in particular, this means that when an element is deleted,
it is simply removed from the table, and no other elements are moved around. Despite both empirical work \cite{originalstability, sandersstability}
and theoretical work \cite{larson1983analysis} on analyzing such schemes,
the complex dependencies between insertions and deletions over time have prevented any analysis from offering provable guarantees
at high load factors (see discussion in \cite{sandersstability}).\footnote{And even if such guarantees were possible, the performance degradation
\cite{originalstability} that these schemes incur at high load factors would still appear to be problematic for proving time bounds on unsuccessful searches.}

\paragraph{Our technique for stability: an unmanaged backyard.} 
A trademark of the use of backyards in recent work \cite{liu2020succinct, arbitman2010backyard, Arbitman09Deamortized, dietzfelbinger1990new, demaine2005dynamic, goodrich2011fully, goodrich2012cache} has been the design of creative ways to move elements
from the backyard to the front yard whenever space frees up in the latter (for example,
Arbitman et al. \cite{arbitman2010backyard} store the backyard as a deamortized
cuckoo hash table, and whenever a cuckoo eviction is performed, they check whether the element
can instead be moved back to the front yard). 

Of course, another approach would be to simply leave the backyard
\emph{unmanaged}, allowing for elements to remain in the backyard \emph{even when
space frees up in the frontyard}. We prove a general-purpose result that we 
call the \defn{Iceberg Lemma}, which establishes that backyards
do not, in fact, require any maintenance to stay small.\footnote{The name of the lemma stems from the fact that the 
majority of an iceberg remains naturally underwater, while only a small portion protrudes above.} 
An essential ingredient of the Iceberg Lemma is that it bounds the size of the backyard not just with high probability, but also with super high probability (in fact, probability $1 - 1 / 2^{n / \polylog n}$). This ends up being central to our data-structure design, as it allows for stability and super-high probability guarantees to be achieved simultaneously, without being at odds with one another.

The approach of having an unmanaged backyard is analogous to the use of tombstones in open addressing. In both cases,
one takes a data structure in which one would normally move elements around and one simply 
analyzes what happens if instead elements are always left in place. The result is that there 
are intricate circular dependencies between where elements reside over time, 
depending on the history of past insertions, deletions, and re-insertions. To overcome these dependencies and achieve super-high probability guarantees, our proof of the Iceberg Lemma makes use of a number of interesting combinatorial ideas. 

\paragraph{\boldmath Using only $O(\log \log n)$ extra bits per key.}
The first hash table to achieve constant-time operations with a load factor of $1 - o(1)$ was that of Arbitman et al.~\cite{arbitman2010backyard}. They achieve a load factor of $1-\varepsilon$, where 
$\varepsilon = O(\sqrt{\log \log n} / \sqrt{\log n})$.  The same paper poses as an open question whether a smaller $\varepsilon$ is achievable. 
Recently, Liu, et al.~\cite{liu2020succinct} presented the first progress on this problem, shaving a $\sqrt{\log\log n}$ factor, and achieving $\varepsilon = O(1/\sqrt{\log n})$.\footnote{Although~\cite{liu2020succinct} considers only insertions (and no deletions), the same basic approach can be made to work with deletions, using the allocate-free version of their techniques (see Section~$7$ of~\cite{liu2020succinct}).}

Iceberg hashing further improves $\epsilon$ to $O(\log \log n / \log n)$. 
This is an especially big improvement in the common case where keys consist of $\Theta(\log n)$ bits.
Here, Iceberg hashing uses only $O(\log \log n)$ extra bits per key in comparison to the previous state-of-the-art of $\Theta(\sqrt{\log n})$ extra bits per key~\cite{liu2020succinct}. 

For $\Theta(\log n)$-bit keys, we also show how to implement Iceberg hashing as a succinct data structure, using only $O(\log \log n)$ extra bits per key when compared to the information-theoretic optimum. In achieving this space bound, our hash table is the first dynamic dictionary to match the lower bound of Demaine et al.~(Theorem 2 of \cite{demaine2005dynamic}, which applies to any static-capacity dynamic dictionary) on the number of bits required by any hash table that stores elements by assigning them stable positions in an array. Thus our hash table has provably optimal space usage across all such hash tables. 

Interestingly, in addition to enabling a stable backyard, the Iceberg Lemma ends up independently
 playing an important role in our high-space-efficiency results. In particular, it allows for 
 the use of backyarding as a way to store metadata succinctly.

\paragraph{In-place dynamic resizing.}
A hash table supports dynamic resizing if the space consumption is a function of the current number of records $n$, rather than some upper bound $N$ on the number of records that could ever be in the data structure. 

Arbitman, Naor, and Segev~\cite{arbitman2010backyard} pose the open question of how to maintain a constant-time, space-efficient hash table that  supports dynamic resizing.  Recently, Liu, Yin, and Yu~\cite{liu2020succinct} gave an elegant solution to this problem, in which records are stored in bins and each bin is represented space-efficiently with fine-grained memory allocations, where the bin  is incrementally expanded/contracted by allocating/deallocating small chunks of memory. The resulting layer of indirection is incompatible with $1+o(1)$ cache optimality.

We remark that there are many approaches to resizing hash tables that incur an $\Omega(1/\epsilon)$-factor time overhead in order to maintain a space efficiency of $1 - \epsilon$. For an excellent discussion of such approaches, see, e.g., \cite{maier2019dynamic}.

\paragraph{Our technique for indirection-free resizing: waterfall addressing.} 
Waterfall addressing revisits the most natural approach to maintaining a space-efficient hash table, which is to simply incrementally resize the table by
$1 \pm o(1)$ factors so that it always stays at a high load factor. The problem with this approach, and the reason that it has not been used in past work,
is that each resize na\"{\i}vely requires $\Omega(n)$ work to rebuild the table, making the approach time inefficient. 

Waterfall addressing maps elements to bins in a way that offers the following guarantees.  Whenever the table size increases by a $1 + o(1)$ factor, only a $o(1)$ \emph{fraction} of elements have their bin changed, and in fact, the only elements whose bin change are the ones that move into the newly created portion of the hash table. Moreover, waterfall addressing allows a time-efficient way to identify \emph{which} elements need to be moved, so a resize can be performed in time proportional to the amount by which the table size is changing. Finally, the probability that any element lands in any bin is nearly uniform, both before and after resizing.

\begin{rmk}
At this point, it is worth taking a moment to expand on the subtle relationship between stability, resizing, and constant-time operations. As discussed earlier, stability is a property that holds at all times \emph{except} for when a hash table is being resized.\footnote{Note that this restriction is fundamental. In particular, as elements are removed from a hash table, the remaining elements \emph{must} be moved to occupy a smaller portion of memory (lest we incur poor space utilization).} On the other hand, in order to deamortize resizes (so that every operation takes time $O(1)$), constant-time hash tables spread the resize operation across a sequence of operations. In our constructions, if a resize increases the table size by a $(1 + 1/s)$ factor, then the work for the resize will be spread across $O(n / s)$ operations. This means that stability kicks in only after these next $O(n / s)$ operations are complete. Of course, another natural approach is for each rebuild to occur atomically in $O(n/s)$ time. If one uses this approach, then the only operations that violate stability are those that trigger a rebuild. \end{rmk}

\paragraph{\boldmath $(1 + o(1))$ cache optimality.}
Whereas the standard RAM model evaluates the running time of an algorithm in terms of the number of operations performed, the External Memory (EM) model~\cite{AggarwalVi88IO} measures performance in cache misses (sometimes called block transfers or I/Os). The EM model has two parameters, the size $M$ of the cache and the size $B$ of a cache line (both measured in machine words). 

Any constant-time hash table trivially incurs $O(1)$ cache misses per operation. Jensen and Pagh~\cite{jensen2008optimality} showed that a much stronger guarantee is possible: there is a constant $c$ such that if $M\ge cB$, one can implement a hash table having load factor $1 - O(1/ \sqrt{B})$ and supporting each operation with $1 + O(1 / \sqrt{B})$ expected cache misses (at the cost of some extra computation). 

In the case where $B \le \log^2 n / \log \log n$, Iceberg hashing achieves  nearly as strong a guarantee on cache misses, while also reducing the computational cost to $O(1)$.  Specifically, using a cache of size $M \ge \polylog n$, an Iceberg hash table achieves a load factor of $1 - O(\sqrt{\log B} / \sqrt{B})$ with $1 + O(1 / \sqrt{B})$ expected cache misses per operation.\footnote{The upper bound on $B$ is necessary given that achieving the same space efficiency for larger $B$ would require further improvement on the state of the art for hash-table space efficiency in the RAM model.} 

\paragraph{Very-high probability constant-time guarantees.}
Goodrich et al.~\cite{goodrich2011fully, goodrich2012cache} consider the problem of achieving subpolynomial probabilities of failure in a constant-time hash table. They note that modern hash tables have two sources of failure: failures due to hash functions being not sufficiently random; and failure due to the design of the table itself. Failures of the former type stem from the fact that the best known families of hash functions \cite{pagh2008uniform, dietzfelbinger2003almost, siegel2004universal} make use of expander graphs, the deterministic construction of which remains one of the longest-standing open problems in extremal combinatorics. As noted by Goodrich et al.~\cite{goodrich2011fully, goodrich2012cache}, however, it is nonetheless possible to isolate out failures of the second type by simply assuming access to a fully random hash function. 
Under this assumption, the authors~\cite{goodrich2011fully, goodrich2012cache} construct the first hash table to have a subpolynomial failure probability, specifically achieving a $1/2^{\polylog n}$ probability of failure with a load factor of $1 - \varepsilon$ for an arbitrarily small constant $\varepsilon > 0$.

Iceberg hashing matches this probability guarantee with an interesting twist: if there exists a hash table with a lower failure probability $p$ and that supports a constant load factor, then Iceberg hashing can be automatically improved to have failure probability $O(p) + 2^{-n / \polylog n}$ (and without compromising any of the other guarantees on space efficiency, cache efficiency, dynamic resizing, and stability).

The smallest achievable value of $p$ remains an open question. The hash table of \cite{goodrich2012cache} achieves 
$p = 1/2^{\polylog n}$, which is the state of the art. We show that, in the common case where keys are $\Theta(\log n)$ bits, a substantially smaller failure probability of $p$ is achievable. In particular, we give a simple data structure that achieves failure probability $p = 1/2^{n^{1 - \epsilon}}$ (for a positive constant $\epsilon$ of our choice). This, in turn, implies that the same failure probability
can be achieved for Iceberg hashing in this case.

As we shall discuss later, all of the properties of Iceberg hashing besides very-high-probability guarantees can be implemented using known families of hash functions (assuming the description bits of the hash function are cached). 
The known results on very-high-probability guarantees (including ours) all require access to fully random hash functions (or other families of hash functions that are not yet known to exist). Removing this requirement remains an interesting direction for future work. 

\subsection{Paper Outline}
The rest of the paper proceeds as follows. 

\begin{itemize}[leftmargin=*]
    \itemsep.1em
    \item Section \ref{sec:iceberg-lemma} proves the Iceberg Lemma.
    \item Section \ref{sec:iceberg-hashing} presents a basic version of the Iceberg hash table that
    is space efficient, cache efficient, and stable. Subsequent sections then build on this basic
    data structure to achieve further guarantees.
    \item Section \ref{sec:dynamic} shows how to perform fine-grained dynamic resizing on an
    Iceberg hash table, using waterfall addressing. We show how to achieve space-efficient dynamic resizing without compromising other properties,  such as cache
    performance.
    \item Section \ref{sec:space} extends the parameter range in which Iceberg hashing can
    be implemented in order to allow for further improvements to space efficiency. This results in
    a load factor of $1 - O(\log\log n / \log n)$.
    \item Section~\ref{sec:prob} considers the problem of achieving subpolynomial failure guarantees assuming fully random hash functions. We show how to implement Iceberg hashing for $\Theta(\log n)$-bit
    keys in a way that achieves failure probability $1/2^{n^{1 - \epsilon}}$. 
    \item Section~\ref{sec:quotient} uses quotienting to make Iceberg hashing into a fully succinct
    data structure, meaning that the space consumption is $(1+o(1))$ times the theoretical optimal.  As in 
    past work~\cite{arbitman2010backyard, liu2020succinct}, we focus on the case where keys are $\Theta(\log n)$ bits. This results in the first succinct dynamic hash table
    to support constant-time operations (with high probability) and waste only $\Theta(\log \log n)$
    bits of space per key when compared to the information-theoretical optimum.
    \item Section~\ref{sec:related} presents an overview of related work.
    \item Finally, Appendix~\ref{sec:hashing} gives explicit families of hash functions that can be used
    to implement Iceberg hashing. As is the case for past hash tables, this introduces a $1/\poly n$
    probability of failure. 
\end{itemize}

Since several variations of the Iceberg hash table are presented in the paper, each of which builds on the previous one, the reader may find it helpful to sometimes reference Appendix \ref{app:figures} which includes a table of (1) the different types of metadata used in each section; and (2) the different reasons that an element can end up in the backyard.

\subsection{Applications in Later Work}

Since the preliminary version of this paper, Iceberg hashing and the techniques that it contains have found a variety of applications to both hashing \cite{PandeyBeJo123, kuszmaul2022hash, bender2022linear} and related areas \cite{bender2023tiny, bender2021paging, gosakan2023mosaic}. Although the original purpose of this paper was primarily to answer a set of theoretical questions, the techniques and data structures have subsequently been used in applied settings \cite{PandeyBeJo123, bender2021paging, gosakan2023mosaic}. 

We begin by discussing applications within the theory literature. Bender, Kuszmaul, and Kuszmaul (FOCS'21 \cite{bender2022linear}) study the question of whether it is possible to construct linear-probing hash tables that avoid a phenomenon known as \emph{primary clustering}. The data structure that they introduce, known as \emph{graveyard hashing}, makes use of waterfall addressing (introduced in Section \ref{sec:dynamic}) to perform efficient in-place resizing without compromising I/O guarantees. Kuszmaul (FOCS'22 \cite{kuszmaul2022hash}) studies the question of whether a failure probability of $1 / 2^{n^{1-\epsilon}}$ can be achieved without assuming fully random hash functions. The resulting data structure, called an \emph{Amplified Rotated Trie}, builds upon the basic design/analysis framework in Section \ref{sec:prob} for how to amplify probability guarantees by using a trie as a backyard data structure. Bender et al.~(SODA'23 \cite{bender2023tiny}) introduce a data-structural primitive called the \emph{tiny pointer}, which can be used as a space-efficient swap-in replacement for pointers in many applications. The construction of the tiny pointer relies heavily on the fact that very-high-probability stable backyards are possible (shown in Section \ref{sec:iceberg-lemma}). 

On the applied side, Gosakan et al.~(SPAA'21 \cite{bender2021paging}, ASPLOS'23 Distinguished Paper \cite{gosakan2023mosaic}) propose the use of Iceberg hashing to assign virtual pages to physical page addresses in RAM. Here, the stability of Iceberg hashing is critical, as the virtual-to-physical mapping cannot be changed dynamically. Another property that they exploit is the \emph{low associativity} of Iceberg hashing (i.e., the fact that each item is guaranteed to be in one of a small number of positions, which follows from the techniques used to achieve space-efficient stability in Section \ref{sec:iceberg-hashing}). This, combined with stability, allows for \cite{bender2021paging, gosakan2023mosaic} to redesign a performance-critical piece of hardware known as the translation look-aside buffer (TLB) in order to achieve significant performance improvements in a wide range of experiments \cite{gosakan2023mosaic}.

Pandey et al.~\cite{PandeyBeJo123} implement a practical version of Iceberg hashing, called \texttt{IcebergHT} (SIGMOD'23 \cite{PandeyBeJo123}). \texttt{IcebergHT} is a state-of-the-art space-efficient concurrent hash table designed for persistent memory and RAM.
The stability of Iceberg hashing ends up playing an important role in the design, as it lends itself to simple and efficient concurrency mechanisms. Additionally, the routing-table techniques in Section \ref{sec:iceberg-hashing}, which \emph{a priori} might seem purely theoretical, end up lending themselves naturally to the use of \texttt{AVX-512} vector instructions. The $(1 + o(1))$ cache-efficiency guarantee also comes into play, as it reduces TLB misses (and thus page-table lookups) compared to alternatives such as Cuckoo hashing.  Finally, the space-efficiency of Iceberg hashing allows for \texttt{IcebergHT} to operate continuously at $> 85\%$ full without compromising performance.

 \section{Iceberg Lemma}
\label{sec:iceberg-lemma}

As discussed in the introduction, a common technique for implementing space-efficient hash tables
is to use a front yard data structure for most elements and a backyard data structure on a small subset of overflow elements.  Since the backyard is so small, the 
data structure used to implement it need not be as space efficient.  

This section considers the question of what happens if the backyard is unmanaged, meaning that
once an element is placed into the backyard, it is not moved back to the front yard, even if space frees up
in the appropriate bin. 

\paragraph{The \historyless.} 
We capture the problem formally with what we call the \historyless, which ignores the structure of the backyard but allows us to bound its size.
The \historyless considers $n$ bins and a universe $U$ of balls. A sequence of ball insertions and removals are performed
over time, with the only constraints being that there are never more than $m = hn$ balls in the system
at any given moment, and that all balls in the system are distinct.
Whenever a ball is inserted, it is hashed to a random bin (if the same ball is inserted, deleted, and later reinserted,
the same bin assignment is used). If the bin being inserted into contains more than $h + \tau_h$ balls (including the ball currently being inserted), where $\tau_h$ is a parameter we will set later,
then the ball is labeled as \defn{exposed}.\footnote{One should think of the balls in the $\historyless$ as stacked up inside the bins in the order of arrival, forming an ``iceberg'': balls at height at most $h + \tau_h$ are below the sea level, and all other balls are afloat and exposed. Since balls can be deleted, some exposed balls may sink below sea level, but they will retain the exposed label. In other words, this label doesn't refer to the current location of a ball in the iceberg layout, but rather to whether it was exposed upon insertion.} 
Note that the number of exposed balls at any time is an upper bound for the number of balls that would be in a backyard.

The \historyless takes an intentionally liberal approach to labeling balls as exposed. In particular, it would be natural
to consider exposed balls as residing in a backyard, and thus not counting towards the fills of the bins in the front yard.
On the other hand, in the \historyless, we intentionally count all balls towards the fills of bins. As we shall discuss
in more detail later, this ensures the following useful property:  whether a given ball is exposed or not depends only on the 
set of balls present during the insertion, rather than on the entire history of the system. 

The threshold $\tau_h$ needs to be chosen so that it is small enough to make the resulting hash table space efficient and 
large enough to make the backyard small.  If $h \le \polylog m$, which is the relevant parameter regime for hash tables, 
then will show that 
$$\tau_h = k\cdot (h\log h)^{1/2},$$
for large enough constant $k$, is a good choice.
For the sake of cleaner calculations,
we also define an additional parameter $c = k^2/3$
which we will use in the analysis. 

\paragraph{A simple example. }As an illustrative example, consider a game with $n = 2$ bins, $h = 1$, and $\tau_h = 0$. Consider the following sequence of insertions and deletions, where every ball is assumed to hash to bin 1. First, balls $x_1$ and $x_2$ are inserted, then $x_1$ is deleted, and then $x_3$ is inserted.

Ball $x_1$ is not exposed. Ball $x_2$ is exposed since, when it is inserted, bin $1$ contains $2 > h + \tau_h = 1$ balls (namely, $x_1$ and $x_2$). After $x_1$ is deleted, $x_2$ remains exposed since being exposed is a static property. Then, when ball $x_3$ is inserted, it too is exposed, since the bin still contains $2 \ge  h + \tau_h = 1$ balls (namely, $x_2$ and $x_3$). Note that the fact that $x_2$ was exposed did not affect $x_3$'s exposure.

As discussed above, this definition of exposure is more aggressive than the overflow condition that Iceberg hashing will actually use---in an Iceberg hash table, ball $x_2$ would end up in the backyard, and then ball $x_3$ would end up in the front yard because bin 1 would be empty when $x_3$ was inserted (as $x_2$ would not be in it). This is a case where by stating a slightly stronger result (i.e., bounding exposed balls instead of backyard balls), we obtain a result that is slightly easier to prove.

\paragraph{Bounding the number of exposed balls. }
An adversary that wishes to force a state where almost all balls were exposed should seek to delete non-exposed balls and insert new balls (or reinsert the old ones), hoping that the new ones become exposed. The following lemma shows that it is almost impossible for an oblivious adversary to achieve this goal.  We use the convention that an event happens \defn{with super-high probability (w.s.h.p.) in \boldmath $n$} if it happens with probability at least $1- 2^{-n / \polylog n}$. 

\begin{lem}[Iceberg Lemma]
\label{lem:iceberg-lemma}
As long as $h \le \polylog m$, then at every point in time the number of exposed balls is at most $n / \poly h$ w.s.h.p.\ in $m$.
\end{lem}

In this section we assume $h \leq \polylog m$. This means that \emph{w.s.h.p.\ in $n$} is equivalent to \emph{w.s.h.p.\ in $m$}, so for now on we will simply say \emph{w.s.h.p.}~without specifying the variable. For the same reason, $n / \polylog n = m / \polylog m$ everywhere.

We remark that if  $h > \polylog m$, and we were to set $\tau_h = h^{1/2 + \epsilon}$ for $\epsilon>0$, then a result analogous to the Iceberg Lemma (but only w.h.p. rather than w.s.h.p.) would be immediate because standard Chernoff bounds would show that there are no exposed balls, w.h.p. What makes the $h \le \polylog m$ case interesting is that there are (almost certainly) going to be exposed balls, but we want to show that there will not be too many. 

\subsection*{Proving the Iceberg Lemma}

\paragraph{A simple argument to achieve a w.h.p. bound.}
We begin by observing that there is a very simple argument that can be used to prove a w.h.p.~version (rather than a w.s.h.p~version) of the Iceberg Lemma for any sequence of $\poly n$ ball insertions/removals. The argument goes in three steps. First, we partition the bins into groups of size $n^{\epsilon}$, and argue that, w.h.p., each group always has at most $n^{\epsilon} + n^{2\epsilon/3}$ balls that hash to it at a time; the rest of the argument conditions on a fixed outcome of which group each ball hashes to. Second, we use linearity of expectation to bound the expected number of exposed balls for each group of bins. Finally, we use the fact that, once we have conditioned on which balls hash to which groups, the numbers $r_1, r_2, \ldots, r_{n^{1 - \epsilon}}$ of exposed balls that are in each of the $n^{1 - \epsilon}$ groups are independent random variables with values in the range $[0, O(n^{\epsilon})]$; applying Hoeffding's inequality, we can conclude that the total number $X = \sum_i r_i$ of exposed balls is tightly concentrated around its mean, w.h.p.

This basic technique of breaking the bins into groups, conditioning on how many balls hash to each group, and analyzing the groups independently, is a classic approach for handling dependencies in balls-and-bins games (and has been used, for example, to handle limited independence in hash tables that require high independence \cite{arbitman2010backyard, liu2020succinct} and to construct quotient-friendly families of permutation hash functions \cite{demaine2005dynamic}). The limitation of the technique, however, is that it achieves much weaker probability bounds than the super-high-probability bounds that we want for the Iceberg Lemma.

In order to achieve tight probabilistic bounds, we will need to take a more sophisticated approach that analyzes all of the balls/bins together, and carefully handles the subtle interdependencies between operations in the operation sequence. By allowing for an unmanaged backyard, while also establishing super-high probability bounds, the Iceberg Lemma will allow for us to achieve both stability and super-high-probability guarantees in the Iceberg hash table. It turns out that the strong probability bounds offered by the Iceberg Lemma also enable applications of the lemma to other areas in data structures---we outline a number of such applications in a subsequent paper \cite{bender2023tiny}.

\paragraph{Notation.}
Before we discuss the w.s.h.p.~analysis, let us take a moment to define some notation. Let $t$ be any fixed time step. Let $A = \{a_1, \dots, a_{m'}\}$ be the set of balls present in the system at time $t$. Let $t_i$ be time where $a_i$ was most recently inserted, and let $T = \{t_1, \dots, t_{m'}\}$. 
Let $B$ be the set of balls other than those in $A$ that are present at any $t_i \in T$ and 
denote the balls in $B$ by $b_1, b_2, \ldots, b_{|B|}$.
Observe that $m' \leq m$, since there are there are at most $m$ balls present at time $t$, and that $|B| \leq O(m^2)$, since there are at most $m$ balls present during each time $t_1, t_2, \ldots, t_{m'}$. 

Let $\alpha = (\alpha_1, \dots, \alpha_{m'})$ be the bin choices of the $a_i$'s, and let $\beta = (\beta_1, \dots, \beta_{|B|})$ be the bin choices of the $b_i$'s. Finally, let $X_i$ be the indicator variable that is $1$ exactly if ball $a_i$ is exposed at its insertion time $t_i$. Then $X = \sum_i X_i$ is the number of exposed balls at time $t$. We want to prove that $X < n / \poly h$ w.s.h.p.

\paragraph{The difficulty of performing a tight probabilistic analysis on \boldmath $X$.}
Roughly speaking, the main challenge in the analysis stems from the fact that each ball $a_i$ may be present during an arbitrary subset of past time steps in $T$, since balls can be deleted and reinserted. This means that the state of the system at steps before $t_i$ may already depend on the randomness of $a_i$. 

Since $X_i$ depends on the balls present at time $t_i$, it follows that $X_i$ depends on the bin choice of $a_j$ for every $j$ such that $t_j < t_i$, and thus $X_i$ depends on $X_j$. On the other hand, because $a_i$ may have been present at time $t_j$ (before being removed and subsequently reinserted at $t_i$), $X_j$ may also depend on $X_i$. In particular, this latter type of dependency implies that we cannot treat $a_i$ as choosing a bin uniformly and independently at random at time $t_i$.

Given that the $X_i$'s are not independent, it is natural to hope that they might nonetheless be stochastically dominated by a sum of independent $0$-$1$ random variables $Y_1, \dots, Y_{m'}$. In particular, one can show that, w.s.h.p., at every time step $t_1, \dots, t_{m'}$ there is at most a small fraction $p$ of bins that have load above $h + \tau_h$. This suggests that, perhaps, the $X_i$'s should be stochastically dominated by independent $Y_i$'s each with mean $p$.

Perhaps surprisingly, this stochastic-dominance approach does not work (even w.h.p.), as one can see with the following example, which highlights some of the subtle dependencies between $X_i$'s. 
Consider the basic setting in which $m = 2$ and balls are labeled as exposed if they land on top of another ball (note that, in this case, we have $p = 1/n$). The adversary performs the following sequence of operations on two balls, $a$ and $b$: (1) insert $a$, (2) insert $b$, (3) delete $a$, (4) insert $a$. Let $X_1$ indicate whether $b$ is exposed at step 2, and let $X_2$ indicate whether $a$ is exposed at step 4. Both of $X_1$ and $X_2$ are $1$ exactly when $a$ and $b$ choose the same bin, and thus $X_1 = X_2$ deterministically. Since $\Pr[X_1 + X_2 = 2] = p$, the random variables are not dominated by independent random variables $Y_1, Y_2$ with mean $p$. This example can be  extended to an arbitrary $m$ and threshold $h + \tau_h$, by  adding $m - 2$ redundant balls before step 1, and then replacing them with another $m - 2$ balls before step 3; note that, in general, this does not result in $X_1 = X_2$, but instead in a subtle positive dependence between $X_1$ and $X_2$.
\footnote{We caution that these dependencies can be quite tricky to reason about. For example, in past work there are several examples where authors attempted to use an unmanaged backyard \cite{Bercea2020Dictionary, BenderFaGo18}, and either incorrectly assumed independence between $X_i$s \cite{BenderFaGo18}, or attempted to perform an erroneous stochastic-dominance argument as described above \cite{Bercea2020Dictionary}---fortunately, this issue is not a big deal in either case, since in both cases it is straightforward to manage the backyard in question in order to fully recover the claimed results.}

\paragraph{Using McDiarmid's Inequality.}
In order to prove the Iceberg lemma, we first discuss a useful inequality: 

\begin{thm}[McDiarmid's inequality \cite{McDiarmid89}]
Let $X_1, \dots, X_k$ be independent random variables taking values from an arbitrary universe $U$. Let $F: U^k \to \mathbb{R}$. Suppose $F$ satisfies the following \defn{Lipschitz condition}: there exists a real number $\ell$ (the \defn{Lipschitz bound}), such that for all $i \in [k]$, $x_1, \dots, x_n, \hat{x_i} \in U$,
\[|F(x_1, \dots, x_i, \dots, x_k) - F(x_1, \dots, \hat{x_i}, \dots, x_k)| \leq \ell.\] 
Let $X = F(X_1, \dots, X_k)$. Then, for all $b > 0$,
\[\prob{X \geq \expect{X} + b} \leq \exp\paren{-\frac{2b^2}{k\ell^2}}.\]
In particular, if $\ell \le \polylog k$, then we can conclude that $X < \expect{X} + k / \polylog k$ w.s.h.p.\ in $k$.\footnote{Indeed, setting $\ell \le \polylog k$ and $b = k / \polylog k$ gives a probability of the form $\exp\paren{-\frac{2b^2}{k\ell^2}} = \exp\paren{-\frac{k^2/\polylog^2 k}{k\polylog k}} = \exp\paren{-k / \polylog k}$.}
\label{thm:mc}
\end{thm}

What happens if we try to apply McDiarmid's inequality to the random variable $X$ as a function $F$ of the $\alpha_i$'s and $\beta_i$'s? 
This comes with two issues: the first is that there are up to $\Theta(m^2)$ different $\beta_i$'s, meaning that $k$ is $\Theta(m^2)$ (which 
is too large to be useful); and the second is that the Lipschitz condition ends up not being satisfied (although, as we shall see, it is
``close'' to satisfied). 

\paragraph{A two-phased analysis.}
To enable the use of McDiarmid's inequality, we will break the analysis into two phases. In the first phase, we will consider an arbitrary $\beta$ and analyze the random variable $X \mid \beta$, that is, the random variable $X$ in which we are using a predetermined $\beta$ (so the only remaining randomness is in $\alpha$). Since $\alpha$ has dimension only $O(m)$, (with a few additional ideas) we can use McDiarmid's inequality to show that, for any fixed $\beta$, $X \mid \beta$ is tightly concentrated around its mean $\E[X \mid \beta]$. 

The second phase of the analysis will then bound $\E[X \mid \beta]$ as a random variable that depends on $\beta$'s randomness. 
Although $\E[X \mid \beta]$ is a random variable (as a function of $\beta$), the fact that it is also an expectation (as a function of $\alpha $)
will allow for us to use linearity of expectation in order to avoid any complications having to do with dependencies across time. 
Leveraging this, we will show that $\E[X \mid \beta]$ is tightly concentrated around $\E[X]$. 

Combining together the two phases of the analysis, we will finally be able to conclude that $X$ is tightly concentrated around $\E[X]$.
We now perform the first phase of the analysis.
\begin{clm}
\label{clm:exposed-mcdiarmid}
For any value of $\beta$, the random variable $X \mid \beta$ satisfies
$$X \mid \beta \leq \expect{X \mid \beta} + n/\polylog n$$ w.s.h.p.
\end{clm}
\begin{proof}

How much is the value of $X$ affected when a single $\alpha_i$ changes? The answer is at most the number of balls present at time $t$ that chose either the old or the new value of $\alpha_i$.\footnote{At first glance, the effect of changing a single $\alpha_i$, that is, changing the bin to which ball $a_i$ hashes, would seem to only change $X$ by at most $1$.  But in fact, the effect can be much larger.  Suppose that the new choice of $\alpha_i$ is a bin that has $h+\tau_h -1$ balls in it.  Placing ball $a_i$ in that bin now fills the bin.  Suppose now that there is a sequence of interleaved insertions of new balls and deletions of unexposed balls from this bin.  With ball $a_i$ in this bin, all the new balls are exposed, and had $a_i$ not been in the bin, none of the new balls would have been exposed (because the alternating insertions and deletion would have kept the bin just under capacity).  So the worst-case effect on $X$ of changing $\alpha_i$ could be as large as $m$.}
Unfortunately, there could be as many as $m$ such balls, meaning we cannot directly apply McDiarmid's inequality.

Define $Y_i = \sum_{\alpha_j = i} X_j$ to be the number of exposed balls in bin $i$ at time $t$, so $X = \sum_i Y_i$. Take $q$ to be a sufficiently large constant and consider 
\[X' = \sum_i \min\left(Y_i, \log^q m\right),\]
that is, a truncated version of $X$ where each bin can contribute at most $\polylog m$ balls. The variable $X'\mid\beta$ is a function of $\alpha$ with Lipschitz bound $\ell = \log^q m$. Hence, by McDiarmid's inequality, $X'\mid\beta \leq \expect{X'\mid\beta} + n / \polylog n$, w.s.h.p..

To complete the proof, we show that, w.s.h.p.,
\begin{equation}
    X \mid \beta \le X' \mid \beta + n / \polylog n.
    \label{eq:W}
\end{equation}
In particular, this would mean that, w.s.h.p.,
\begin{align*}
 X \mid \beta & \le X'\mid\beta  + n / \polylog n \\
    &\leq \expect{X'\mid\beta} + n / \polylog n \\
    & \leq \expect{X\mid\beta} + n / \polylog n.
\end{align*}
We now prove \eqref{eq:W}.
For each bin $ i $, define $W_i$ to be 
$$W_i = \max\left(0, |\{j \mid \alpha_j = i\}| - \log^q m\right),$$
that is, if more than $\log^q m$ balls $\{a_j\}$ land in bin $i$, then $W_i$ counts the number of excess balls. By design, $X - X' \le \sum_i W_i$ deterministically. 
Notice, however, that $\sum_i W_i$ is a function of the $m$ independent random variables $\alpha = \{\alpha_i\}$ with Lipschitz bound $\ell = 1$, and that $\E[\sum_i W_i] = o(1)$ (since by a Chernoff bound each $W_i = 0$ w.h.p.). Thus we can apply McDiarmid's inequality 
to deduce that $\Pr[\sum_i W_i \mid \beta > n / \polylog n] \le 2^{-n / \polylog n}$, completing the proof.
\end{proof}

We next turn to the second phase of the analysis, which is to prove a concentration bound on the random variable $\E[X \mid \beta]$ (whose outcome depends only on the randomness in $\beta$). 
Say that a bin is \defn{heavy} (at a given point in time) if it contains at least $h + \tau_h$ balls. 
Say that a time step is \defn{bad} if there are more than $2n / h^c$ heavy bins at that step, and otherwise we say it's \defn{good}. The reason for these names is that, if a ball is inserted during a good step, then its probability of being exposed is at most $\frac{2n / h^c}{n} = 2 / h^{c}$.

Let $B_i$ be the indicator random variable that is $1$ if and only if $t_i$ is a bad step, and let $B = \sum_i B_i$. 

\begin{clm}
\label{clm:expect-bad-steps}
For any choice of $\beta$, we deterministically have that
$$\expect{X \mid \beta} \leq 2n/h^{c - 1} + \expect{B \mid \beta}.$$
\end{clm}
\begin{proof}
Let $\calF_i$ be the event that $a_i$ is labeled exposed at $t_i$. By linearity of expectation, we have

$$\expect{X \mid \beta}  = \sum_i \expect{X_i \mid \beta} = \sum_i \prob{\calF_i \mid \beta}.$$
By considering whether each step is good or bad, we can decompose this as
    \begin{align*}
 &\sum_i \paren{\prob{\text{$t_i$ good} \mid \beta} \cdot \prob{\calF_i \mid \beta, \text{$t_i$ good}} + \prob{\text{$t_i$ bad} \mid \beta} \cdot \prob{\calF_i \mid \beta, \text{$t_i$ bad}}}\\
\leq & \sum_i \paren{\prob{\calF_i \mid \beta, \text{$t_i$ good}} + \prob{\text{$t_i$ bad} \mid \beta}}\\
\leq & \sum_i \paren{2/h^c + \prob{\text{$t_i$ bad} \mid \beta}} \\
= & 2n/h^{c - 1} + \expect{B \mid \beta}. 
\end{align*}
\end{proof}

The following claim shows that w.s.h.p. there are no bad steps.

\begin{clm}
\label{clm:good-step}
Any fixed time step is good w.s.h.p. (with probability taken over both $\alpha$ and $\beta$).
\end{clm}
\begin{proof}
Let $L_j$ be the load of bin $j$ at the fixed step. This is a binomial random variable with mean at most $h$. Let $\eps = (3c \log h)^{1/2} / h^{1/2}$. Then,
\begin{align*}
    \prob{L_j \geq h + \tau_h} &= \prob{L_j \geq h + (3c \log h)^{1/2} h^{1/2}} \tag{by the choice of $\tau_h$}\\
    &= \prob{L_j \geq (1 + \eps)h}\\
    &\leq \prob{L_j \geq (1 + \eps)\expect{L_j}} \tag{as $\expect{L_j} \leq h$}\\
    &\leq \exp\paren{-\frac{\eps^2 h}{3}} \tag{by a Chernoff bound, and since $\expect{L_j} \leq h$}\\
    &= \exp(-c \log h)\\
    &=h^{-c} \numberthis \label{eq:prob-heavy}.
\end{align*}

Let $Z_i$ be the indicator variable that is $1$ if and only if when $L_j \geq h + \tau_h$. Then, $Z = \sum_j Z_j$ is the number of heavy bins. By linearity of expectation and \Cref{eq:prob-heavy}, $\expect{Z} \leq n/h^c$. Since $n/h^c \ge n / \polylog n$, and since the $Z_i$s are negatively associated, a Chernoff bound implies that $Z \leq 2n/h^c$ w.s.h.p.
\end{proof}

We now use \Cref{clm:good-step} to bound $\expect{B \mid \beta}$, as follows.

\begin{clm}
\label{clm:bad-steps}
$\expect{B \mid \beta} = 0$ w.s.h.p. (with randomness taken over $\beta$).
\end{clm}
\begin{proof}
We have that
\begin{align*}
    \expect{\expect{B \mid \beta}} &= \expect{B} \tag{by the tower rule}\\
    &= \sum_i \expect{B_i} \tag{by linearity}\\
    &\leq \sum_i 1/ 2^{n / \polylog n} \tag{by \Cref{clm:good-step}}\\
    &= 1/2^{n / \polylog n}.
\end{align*}
Thus, by Markov's inequality, $\prob{\expect{B \mid \beta} \geq 1} \leq \expect{\expect{B \mid \beta}} \leq 1 / 2^{n / \polylog n}$.
\end{proof}

Finally we put the two phases together to complete the proof.
\begin{proof}[Proof of Lemma \ref{lem:iceberg-lemma}]
We have, w.s.h.p.,
\begin{align*}
    X &\leq \expect{X \mid \beta} + n/\polylog n \tag{by \Cref{clm:exposed-mcdiarmid}}\\
    &\leq 2n/h^{c - 1} + \expect{B \mid \beta} + n/\polylog n \tag{by \Cref{clm:expect-bad-steps}}\\
    &\leq 2n/h^{c - 1} + n/\polylog n \tag{by \Cref{clm:bad-steps}}\\
    &\leq n / \poly h.
\end{align*}
\end{proof}

 \section{Basic Iceberg Hashing}
\label{sec:iceberg-hashing}

In this section, we consider the problem of constructing a space-efficient hash table
with constant-time operations, referential stablility, and nearly optimal cache
behavior. Our solution is the most basic version of an Iceberg hash table; in subsequent sections, we will show how to modify the table to achieve stronger guarantees.

The lemmas in this section will assume access to constant-time fully random hash
functions; we discuss the use of explicit families of hash functions at the end of the section.

\paragraph{The structure of an Iceberg hash table.}
Let $N$ be an upper bound on the current number of keys $n$ in the table, let $U$ be the
universe of keys, and let $h$ be a parameter satisfying $h \le O(\log
N / \log \log N)$; we call $h$ the \defn{average-bin-fill parameter}.  Let
$\calT$ be an arbitrary hash table implementation that supports
constant-time operations (w.h.p.) and load factor at least
$\frac{1}{\poly h}$ (i.e, the table can store $n$ records in space $n \poly h$). We can further assume without loss of generality that $\mathcal{T}$ is stable, since any hash table can be made stable by adding an extra level of indirection to the records (at the cost of a constant-factor loss in load factor and an extra cache miss per operation).  For now, the specifics of $\calT$'s implementation will be unimportant---this will change later on, in Section \ref{sec:prob}, when we show how to achieve w.s.h.p.~guarantees. 

The Iceberg hash table consists of a front yard and a backyard.  The front yard consists of $N /
h$ bins, each of which has capacity $h + \tau_h$ (recall from Section \ref{sec:iceberg-lemma} that $\tau_h = k \cdot (h\log{h})^{1/2}$ for some constant $k$). The backyard, which will
store only a small number of records (roughly $N / \poly h$) is implemented
using $\mathcal{T}$.

The front yard uses two hash functions: the function $\g: U
\rightarrow [N / h]$ maps keys to bins, and the function $\f: U \rightarrow
[\poly h]$ maps keys to random $\Theta(\log h)$-bit \defn{fingerprints}.

When a new key $x$ is placed into the table, we first try to place it
into its front-yard $\g(x)$. If the $\g(x)$ contains fewer than
$h + \tau_h$ records, and all of the records $y$ in the bin satisfy
$\f(x) \neq \f(y)$, then $x$ is placed into the bin.\footnote{For
  convenience of notation, we will often treat a record as a key
  (rather than a key-value pair), allowing for us to, for example,
  talk about the fingerprint $\f(x)$ for a record $x$.} Otherwise, $x$
is placed into the backyard $\mathcal{T}$.

The insertion procedure ensures that, within each bin, the records all have
distinct fingerprints. This enables a simple space-efficient scheme for
performing queries within the bin. Define a \defn{routing table} to be a
dictionary that maps up to $h + \tau_h$ different fingerprints to indices $i
\in [h + \tau_h]$ within a bin. As we shall discuss shortly, as long as $h \le
O(\log n / \log \log n)$, a routing table can be encoded in $O(1)$ machine
words (and $O(h \log h)$ bits) with constant-time query/insert/delete
operations (and with no cache misses beyond those needed to load the $O(1)$ machine words). The routing table is used within each
bin to map the fingerprint $\f(x)$ of each key $x$ to the corresponding position
of $x$ in the bin.

Each bin $b$ also maintains several other pieces of metadata: a \defn{fill
counter} keeping track of the number of records in the bin, a \defn{vacancy
bitmap} keeping track of which slots are vacant in the bin, and a
\defn{floating counter} keeping track of how many keys $x$ in the backyard
satisfy $\g(x) = b$.

The fill counter and vacancy bitmaps are used to implement insertions in
constant time. The floating counter, on the other hand, is used to make queries
more cache efficient. If a query for a key $x$ goes to a bin $b$ whose floating
counter is $0$, then the query need not search for $x$ in the backyard. As
long as the bin fits in a single cache line, then the query to $x$ incurs only
a single cache miss.

In summary, insertions, queries and deletions work as follows:
\begin{itemize}
    \item An insertion of a record $x$ examines the routing table in $\g(x)$ to determine whether $\f(x) = \f(y)$ for some $y$ in the bin. It then examines the fill counter to assess the number of free slots in the bin. The record $x$ goes to the frontyard if $\g(x)$ contains fewer than $h + \tau_h$ records, and all of the records $y$ in the bin satisfy $\f(x) \neq \f(y)$; and it goes to the backyard otherwise. If $x$ goes to the frontyard, then the vacancy bitmap is used to find a slot in which to place it, and then the metadata (vacancy bitmap, fill counter, routing table) is updated. If $x$ goes to the backyard, then the floating counter is updated, and the backyard hash table $\calT$ is used. 
    \item A query to a record $x$ examines the routing table in $\g(x)$ to determine whether $\f(x) = \f(y)$ for some record $y$ in the bin. If so, the query checks if that record is $x$. Otherwise (if either there is no such $y$ or if $y \neq x$), then the query examines the floating counter for the bin. If the floating counter is $0$, the query terminates (with a negative result). Otherwise, the query uses the backyard table $\calT$.
\item A deletion of a record $x$ first performs a query to find the record. If the record is in the backyard, it is deleted from $\calT$ and the floating counter of the appropriate frontyard bin is updated. If the record is in the frontyard, it is deleted, and the metadata (vacancy bitmap, fill counter, routing table) is updated.
\end{itemize}
\noindent
In \Cref{fig:iceberg-basic} we present a diagram with the basic structure of an Iceberg hash table.

\begin{figure}[H]
\centering
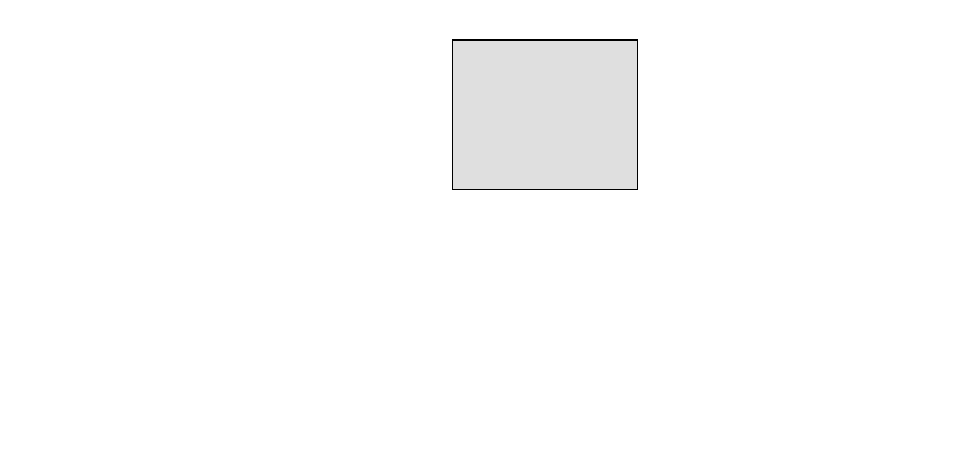
\caption{The structure of the basic Iceberg hash table. All operations on a key $x$ are divided into three steps: First, select the appropriate front yard $\bin(x)$. Second, select a fingerprint slot in the bin, using the bin metadata. And third, run the operation on the backyard table $\mathcal{T}$ if necessary.}
\label{fig:iceberg-basic}
\end{figure}

In order to complete the description and analysis of Iceberg hash tables, we have two main tasks: to show that routing tables can be implemented in $O(1)$ machine words with $O(1)$-time operations, which we do via standard bit techniques; and to show that the backyard remains small, even though records are never moved from the backyard to the front yard, which we do via the Iceberg lemma.

\paragraph{Implementing a routing table with \boldmath $O(1)$ machine words.}
Let $a_1, a_2, \ldots, a_r$ be a set of distinct fingerprints stored in a routing table, and let $b_1, b_2, \ldots, b_r$ be the corresponding indices of the fingerprints within the bin -- the routing table maps $a_i$ to $b_i$. Let $A$ be an array storing $a_1, a_2, \ldots, a_r$, and let $B$ be an array storing $b_1, b_2, \ldots, b_r$. Note that $A$ and $B$ can be stored in $O(1)$ machine words using $O(h \log h)$ bits, since $r = O(h)$ and each element in each array is $\Theta(\log h)$ bits. Since $h = O(\log n / \log \log n)$, we have $|A| = |B| = O(h \log h) = O(\log n) = O(\word)$, where $\word$ is the machine word size. The routing table simply stores $A$ and $B$, for a total of $O(1)$ words.

Queries to the routing table face the challenge of determining whether a fingerprint $a$ is in the array $A$, and if so, then the query must also return $b_i$ for the index $i$ such that $a_i = a$. Fortunately, these operations can be implemented in constant time using standard bit techniques.

In the following, we will exploit the fact that several useful word operations can be performed in constant time. In all our word operations, we will operate on small integers stored in a single word.

When thinking of an integer $x$ as a bit string, we will treat it as being right justified, meaning that $x$'s least significant bit is the final bit of the bit string. We denote the concatenation of two bit strings $x$ and $y$ by $x\circ y = x2^{|y|+1} + y$. That is $x\circ y$ is the bits of $x$, followed by a \defn{padding bit} $0$, followed by the bits of $y$ (reading from most to least significant bit).  We say that $a_1, \ldots, a_k$ are \defn{packed into a word} $A$, if $A = a_1 \circ \cdots \circ a_k$,  and we call the bit before $a_i$ the \defn{$i$th padding bit}.

The proofs of Lemmas~\ref{lem:membership}, \ref{lem:findgi}, and \ref{lem:min}  use the following standard set of tools:
\begin{enumerate}[noitemsep]
    \item Given a word, the position of the \defn{least significant $1$-bit} can be computed in $O(1)$ time~\cite{knuth2011art4}.
    \item Given a word, the position of the \defn{most significant $1$-bit} can be computed in $O(1)$ time~\cite{DBLP:conf/stoc/FredmanW90}.
    \item Given a bit string $a$ of at most $(\word/k) - 1$ bits, the word $A  = a \circ \cdots \circ a$ consisting of $k$ copies of $a$ can be computed in $O(1)$ time~\cite{DBLP:conf/stoc/FredmanW90}.
    \item Given two sets of bit strings $\{x_1,\ldots,x_k\}$ and $\{y_1,\ldots,y_k\}$, where $|x_i| = |y_j| \le (\word/k) - 1$ for all $i,j\in [k]$, let $X = x_1 \circ \cdots \circ x_k$ and $Y = y_1 \circ\cdots\circ y_k$, and let $p_i$ be the location of the $i$th padding bit in both $X$ and $Y$.  Then in $O(1)$ time~\cite{DBLP:conf/stoc/FredmanW90}, we can compute a word $Z$ where for $i\in[k]$, the $p_i$th bit of $Z$ is $1$ if $x_i \geq y_i$ and $0$ otherwise (and the bits not corresponding to padding-bit locations $p_i$ are 0). That is, we can compare every $x_i$ and $y_i$ in $O(1)$ time.
\end{enumerate}

\begin{lem}
\label{lem:membership}
 
    Let $x_1, x_2, \ldots, x_r$ be bit strings of length $s\leq (\word/r) - 1$, let $X = x_1\circ\cdots\circ x_r$, and let $y$ be an $s$-bit number.  In constant time, one can determine whether $y \in \{x_1,\ldots,x_r\}$, and for what index $i$ we have $x_i = y$ (if such an $i$ exists).
  
\end{lem}
\begin{proof}  We use the standard techniques outlined above.  We pack $r$ copies of $y$ in a new word $Y$.  Compare $X$ with $Y$, compare $Y$ with $X$, and AND together the resulting comparison-indicator words.  The result yields an equality-indicator word $Z$ (that is, the $i$th padding bit of $Z$ indicates whether $x_i = y$). We find which $x_i = y$, if any exists, by finding the least significant $1$-bit of $Z$.
\end{proof}

The simplicity of the routing table's encoding makes insertions and deletions of fingerprints easy to implement in constant time. In particular, insertions and deletions simply need to update a single entry in each of the arrays $A$ and $B$.

\paragraph{Analysis of Iceberg hashing.}
The challenge in analyzing the Iceberg hash table is to bound the number of records
in the backyard. There are two types of keys $x$ in the backyard: keys
$x$ that were placed in the backyard due to lack of space in $\g(x)$,
and keys $x$ that were placed in the backyard due to a fingerprint
collision with another key $y$ in $\g(x)$. We refer to keys of the
former type as \defn{capacity floaters} and keys of the latter type as
\defn{fingerprint floaters}. As we shall see, the number of capacity floaters
can be bounded by the Iceberg Lemma, and the number of fingerprint floaters can
be bounded by an analysis using McDiarmid's inequality. 

Although for now we are only interested in $h \le O(\log N / \log \log N)$, later in the paper we will also consider even more space efficient variants of Iceberg hashing in which $h$ is larger. To simplify discussion later, we state several of the lemmas in this section for arbitrary $h \le \polylog N$.

We begin by bounding the number of fingerprint floaters.
We remark that the proof of the next lemma requires a bit of care to avoid any potential
subtle circular dependencies between the random variables being analyzed.
\begin{lem}
  Suppose $h \le \polylog N$. Then w.s.h.p.\  in $N$, there are at
  most $N / \poly h$ fingerprint floaters. Moreover, for a given
  key $x$, the probability that there is a fingerprint floater $y$
  such that $\g(x) = \g(y)$ is at most $1 / \poly h$.
  \label{lem:fingerprints}
\end{lem}
\begin{proof}
  Let $t$ denote the current time, and let $X$ denote the set of keys
  present at time $t$. For each key $x \in X$, let $A_x$ denote
  the event that $(\g(x), \f(x)) = (\g(y), \f(y))$ for some $y \in X\setminus \{x\}$, and
  let $B_x$ denote the event that $(\g(x), \f(x)) = (\g(y), \f(y))$ for
  some $y \not\in X$ such that $y$ was present in the table when $x$
  was inserted. The total number of fingerprint floaters is upper bounded by
  $$\sum_{x \in X} A_x + \sum_{x \in X} B_x,$$
  where $A_x$ and $B_x$ are treated as indicator random variables.

  The function $A = \sum_{x \in X} A_x$ is determined by the
  $|X| \le N$ independent random variables
  $\{(\g(x), \f(x)) \mid x \in X\}$, and $A$ has Lipschitz
  bound $\ell = 1$. By McDiarmid's
  inequality, it follows that $A \le \expect{A} + N / \polylog N$,
  w.s.h.p in $N$. On the other hand, since every $x, x' \in X$ collide in their bin-choice/fingerprint 
  with probability $\frac{h}{N} \cdot \frac{1}{\poly h}$, each event $A_x$
  occurs with probability at most $1 / \poly h$, which means
  that $\E[A] \le N / \poly h$. Thus, w.s.h.p.\ in $N$, we have
  $A \le N / \poly h + N / \polylog N \le N / \poly h$.

  Next we analyze $B = \sum_{x \in X} B_x$. Let $Z_y$ be the outcome of $(\g(y),\f(y))$ and let $Z = \{Z_y \mid y\not\in X\}$.   If we condition on any
  fixed $Z$, then
  the random variables $\{B_x \mid x \in X\}$ become independent. It
  follows by a Chernoff bound that
  $B\mid Z \le \expect{B \mid Z} + N / \polylog N$, w.s.h.p in
  $N$. Define $T_x$ to be the set of elements present when a given $x$ is inserted.  No matter what the outcome of $Z$ is, we have by a union bound that
  \begin{align*}
      \E[B \mid Z] & \le \sum_{x \in X} \sum_{\substack{y \not\in X\\y \in T_x}} \Pr[(\g(x), \f(x)) = (\g(y), \f(y)) \mid Z] \\
      & =  \sum_{x \in X} \sum_{\substack{y \not\in X\\y \in T_x}} \frac{1}{N \poly h} \le N / \poly h. 
  \end{align*}
  Thus we have $B \le N / \poly h + N / \polylog N=  N / \poly h$, w.s.h.p. in $N$.

  So far we have shown that, w.s.h.p.\ in $N$, there are
  at most $N / \poly h$ fingerprint floaters at time $t$. It
  remains to show that, for a given key $x$, the probability that
  there is a fingerprint floater $y$ at time $t$ such that
  $\g(x) = \g(y)$ is at most $1 / \poly h$.

  The probability that $x$ itself is
  a fingerprint floater is at most $1 / \poly h$, and similarly
  the probability that there is any $y$ at time $t$ such that
  $(\g(y), \f(y)) = (\g(x), \f(x))$ is at most $1 / \poly h$. To
  complete the proof, we must bound the probability that there exists
  a fingerprint floater $y$ at time $t$ such that $\g(x) = \g(y)$ and
  such that, when $y$ was inserted, there was a record $z \neq x$
  present such that $(\g(y), \f(y)) = (\g(z), \f(z))$. We know that, w.s.h.p.\ 
 in $N$, there are at most $N / \poly h$ keys
  $y$ such that when $y$ was inserted, there was a record $z \neq x$
  present such that $(\g(y), \f(y)) = (\g(z), \f(z))$. Record $x$ has
  probability $h / N$ of satisfying $\g(x) = \g(y)$ for each of
  these $y$'s. By a union bound, the probability of $x$ satisfying
  $\g(x) = \g(y)$ for any such $y$ is at most
  $O\left(\frac{h}{N} \cdot \frac{N}{\poly h}\right) = 1 /
  \poly h$, which completes the proof. 
\end{proof}

Say a record $x$ is a \defn{capacity exposer} if, when $x$
was inserted, there were already at least $h + \tau_h$ records $y$
present (including the records in the backyard) such that $\g(x) =
\g(y)$. Rather than analyzing the number of capacity floaters directly,
we instead analyze the number of capacity exposers (although this distinction is not important now, it will be later
in our analysis in Section \ref{sec:dynamic}).
\begin{lem}
  Suppose $h \le \polylog N$. Then, w.s.h.p.\ in $N$, there are at
  most $N / \poly h$ capacity exposers in the table.
  Moreover, for a given key $x$, the probability that there is a
  capacity exposer $y$ such that $\g(y) = \g(x)$ is at most
  $1 / \poly h$.
 
  \label{lem:capacity_floaters}
\end{lem}
\begin{proof}
  By the Iceberg Lemma, the number of
  capacity exposers at any given moment is at most $N / \poly h$
  w.s.h.p.\ in $N$.

  It remains to show that, for a given key $x$, the probability of 
  $\g(x)$ having a capacity exposer is at most
  $1 / \poly h$. Let $A$ denote the set of keys $y$ present at
  time $t$ such that, when $y$ was inserted there were at least
  $h + \tau_h - 1$ other keys $z$ in $\g(y)$ satisfying $z \neq
  x$. By the Iceberg Lemma (applied using $\tau_h' < \tau_h - 1$),
  $|A|$ is at most $N / \poly h$ w.s.h.p.\ in $N$. The probability that $\g(x)$ contains any
  elements from $A$ is therefore $1 / \poly h$. On the
  other hand, in order for $\g(x)$ to have a capacity floater
  $y \neq x$, we must have $y \in A$. Thus the probability of  
  $\g(x)$ having a capacity floater is at most
  $1 / \poly h$, completing the proof.
\end{proof}

Combining the preceding lemmas, we analyze the backyard.
\begin{lem}
  Suppose $h \le \polylog N$. Then, w.s.h.p.\  in $N$, there are at
  most $N / \poly h$ records in the backyard.  Moreover, for a
  given key $x$, the probability that $\g(x)$ has a non-zero
  floating counter is at most $1 / \poly h$.
  \label{lem:second_level}
\end{lem}
\begin{proof}
  Let $t$ be the current time. By Lemma \ref{lem:fingerprints}, w.s.h.p.\ 
  in $N$, there are at most $N / \poly h$
  fingerprint floaters at time $t$. Also by Lemma
  \ref{lem:fingerprints}, the probability that a given record $x$
  hashes to a $\g(x)$ for which there is at least one fingerprint
  floater is at most $1 / \poly h$. 
  
  Since every capacity floater is a capacity exposer, we can use Lemma \ref{lem:capacity_floaters}
  to deduce that, w.s.h.p., there are at most $N / \poly h$
  capacity floaters at time $t$. Also by Lemma
  \ref{lem:capacity_floaters}, the probability that a given record $x$
  hashes to a $\g(x)$ for which there is at least one capacity
  floater is at most $1 / \poly h$. This completes the proof.
\end{proof}

The previous lemmas all assume access to fully random hash functions.
In Appendix \ref{sec:hashing}, we describe how to modify Iceberg hashing (both the simple version described in this section and the 
stronger variants in subsequent sections)
to be compatible with an explicit family of hash functions (the transformation
is essentially the same as the one used in past works \cite{arbitman2010backyard, liu2020succinct}).
The transformation preserves all of the properties of Iceberg hashing that we care about (time efficiency,
cache efficiency, space efficiency, and stability) but, as in previous work, this introduces an additional $1/\poly N$ failure
probability due to the hash functions themselves. Thus, in order so that our analysis of Iceberg hashing
is compatible with an explicit family of hash functions, we state Theorem \ref{thm:static} (as well as the other main theorems of the paper) in terms
of w.h.p.\ guarantees rather than in terms of w.s.h.p.\ guarantees. The only exception to this will be in Section \ref{sec:prob} where
we prove w.s.h.p.\ guarantees assuming fully random hash functions, building on past work \cite{goodrich2011fully, goodrich2012cache} which
has also assumed full randomness for the same reasons.

We now present the full analysis of the Iceberg hash table. 
\begin{thm}
Consider an Iceberg hash table that never contains
more than $N$ elements and suppose that the average-bin-fill parameter satisfies  $h = O(\log N / \log \log N)$. Suppose that the
    backyard table supports
    constant-time operations (w.h.p.\ in $n$), supports load factor
    at least $1 / \poly h$, and is stable. 

  Consider a sequence of operations in which the number of records in
  the table never exceeds $N$, and consider a query, insert, or delete
  that is performed on some key $x$. Then, the
  following guarantees hold.
  \begin{itemize}
  \item \textbf{Time Efficiency.} The
    operation on $x$ takes constant time in the RAM model, w.h.p.\  in $N$,.
  \item \textbf{Cache Efficiency.} Consider the EM model using a cache line of size $B \geq 2h$ and a cache of size $M = \Omega(B)$, and suppose that each bin in the front yard is memory aligned, that is, each bin is stored in a single cache line. Finally, suppose that the description bits of the hash functions are cached.
    Then the operation on $x$ has probability at least $1 - 1 / \poly B$ of
    incurring only a single cache miss.
  \item \textbf{Space Efficiency.} The total
    space in machine words consumed by the table is, w.h.p.\ in $N$,
    $$\left(1 + O\left(\frac{\sqrt{\log h}}{\sqrt{h}}\right)\right) N = (1 + o(1))N.$$ 
  \item \textbf{Stability.} The hash table is stable.
  \end{itemize}
  \label{thm:static}
\end{thm}
\begin{proof}
  The claims of time efficiency and stability follow directly from the
  construction of the Iceberg hash table.
  By Lemma \ref{lem:second_level}, the space consumed by the backyard is $N / \poly h$ machine words w.h.p.\ in $N$. The space in machine words consumed by the front yard is deterministically
  $$\frac{N}{h} (h + \tau_h + O(1)) \le \left(1 + O\left(\frac{\sqrt{\log h}}{\sqrt{h}}\right)\right) N,$$
  which completes the proof of space efficiency.

  Finally, we prove the claim of cache efficiency. If the operation on
  key $x$ is a query, then the probability of incurring more than one cache
  miss is equal to the probability that $\g(x)$ has a non-zero
  floating counter. By Lemma \ref{lem:second_level}, this probability
  is at most $1 / \poly h = 1 / \poly B$. By the same analysis, the probability that a deletion incurs multiple cache misses is also $1 / \poly B$. 

  If the operation is an insertion, then the probability of incurring
  more than one cache miss is equal to the probability that $x$ is
  placed into the backyard. This, in turn, is the probability that
  either (1) there is another record $y$ in $\g(x)$ such that
  $\f(y) = \f(x)$; or (2) there are already $h + \tau_h$ records in 
  $\g(x)$. The probability of (1) is at most $N \cdot \frac{h}{N} \cdot 1 / \poly h = 1 / \poly h$ by a union bound, 
  and the probability of (2) is also $1 / \poly h$ by
  a Chernoff bound.

  Thus, for any operation, the probability of incurring more than one
  cache miss in the EM model is $1 / \poly h = 1 / \poly B$.
\end{proof}

We conclude the section with several remarks on how the theoretical techniques discussed in this section relate to practical implementations of Iceberg hashing.

\begin{rmk}
    It is worth taking a moment to comment on what $h$ and $B$ look like in practice. \texttt{IcebergHT}, which is a high-performance implementation of Iceberg hashing by Pandey et al.~\cite{PandeyBeJo123}, sets $h = 128$ in order to consistently achieve a space efficiency greater than $85\%$. Of course, $h = 128$ means that the analysis also requires $B = 128$, which means that for some of caches on a modern machine (e.g., the L1, L2, and L3 caches, which use 64-byte cache lines), the block-size constraint in Theorem \ref{thm:static} is not met. (For these caches, the cache-miss behavior of \texttt{IcebergHT} is closer to that of Cuckoo hashing, which incurs two misses per query). On the other hand, for caches with large line lengths (including the translation look-aside buffer, which uses lines of size 4KB), the $(1 + o(1))$-cache-optimality guarantee in Theorem \ref{thm:static} does apply. This is one of the sources (but not the only one) of \texttt{IcebergHT}'s performance.
\end{rmk}

\begin{rmk}
    It is tempting to assume that the routing-table techniques in this section would be hopelessly impractical in a real system. However, modern CPUs support vector instructions (such as \texttt{AVX-512}) that, in practice, significantly increase the effective size of a machine word. Making use of this observation, Pandey et al.~\cite{PandeyBeJo123} were able to implement a high-performance version of Iceberg hashing with a routing-table structure remarkably similar to the one described here.
\end{rmk}

 \section{Dynamic Resizing with Waterfall Addressing}\label{sec:dynamic}

In this section, we show how to transform Iceberg hashing into a
dynamically resizable hash table, while preserving the space efficiency,
time efficiency, and cache efficiency of the original data structure
(and also preserving stability during time windows in which the table is not resized).

The core challenges that one encounters when trying to make Iceberg hashing space- and time-efficiently dynamically
resizable are the same as those that arise for any other direct-mapped hash table. To capture this set of challenges formally, this section defines the
 \defn{dynamic bin addressing problem}. We then give an efficient solution to the
this problem, which we call \defn{waterfall addressing},
 and we show how to use waterfall addressing to construct a
dynamically resizable version of Iceberg hashing.

\paragraph{Resizing through partial expansions.}
If one does not care about space efficiency, then the classic approach to
dynamically resizing a hash table is to simply rebuild it whenever its size changes by a constant factor.
On the other hand, if space efficiency is a concern, then the hash table must be resized in smaller increments.  

Suppose that we wish to maintain a hash table at a load factor of
$1 - O(1/s)$ for some power-of-two parameter $s=\omega(1)$.
Perhaps the most natural approach is to grow the table through small \defn{partial expansions}.
Each time that the hash table doubles in size, a total of $s$
partial expansions are performed. We call $s$ the \defn{resize granularity}.  If the hash table initially consists of
$2^a$ bins, then each of the partial expansions increases the number of bins by $2^a / s$,
so that after $s$ partial expansions the number of bins becomes
$2^{a + 1}$. The partial expansions are spread out over time so that
the average load on each bin never changes by a factor of more than $1 \pm O(1 / s) = 1\pm o(1)$. 

\paragraph{The problem: dynamically mapping elements to bins.}
How should we map elements to bins after each partial expansion?
If a hash table consists of $m$ bins, and we have a fully random hash function
$g: U \rightarrow [2^w]$ (for some $w$ satisfying $2^w \gg m$), then the classic approach to mapping
elements $ x\in U$ to bins $[m]$ is to simply use the bin assignment function
\begin{equation}
\text{Bin}_m(x) = g(x) \pmod m.
\label{eq:bintriv}
\end{equation}

The problem with this bin assignment function is that, whenever a partial expansion is performed,
almost all of the records in the hash table will have their bin assignments changed.
This means that, if a partial expansion is performed on a hash table with $n$ elements,
then the expansion will require $\Omega(n)$ time. In contrast,
if we wish to have $O(1)$-time operations, then each partial expansion must take time at most
$O(n / s)$. 

\paragraph{The dynamic bin-addressing problem.}
Let $ U $ be the universe 
and $g:U \rightarrow [2^w]$
a fully random hash function, where $w$ is the number of bits in a machine word and $2^w$ is an upper bound on the number of bins that will ever be in our hash table. 

Define $m_{a, j} = 2^a + j \cdot 2^a / s$ to be the number of bins after the
$j$-th partial expansion in the process of doubling a table from $2^a$ to $2^{a + 1}$ bins.
Define a \defn{bin assignment function} $\bin(a, j, x): [\log s, w - 1] \times [s] \times U \rightarrow [m_{a, j}]$
to be the function that assigns keys $x$ to bins after the $j$-th partial expansion in the process of
doubling a table from $2^a$ to $2^{a + 1}$ bins. As an abuse of notation, 
we also define $\bin(a, 0, x) = \bin(a - 1, s, x)$.

A bin assignment function is a solution to the \defn{dynamic bin-addressing problem} if it satisfies the following three properties.
\begin{itemize}
\item \textbf{The Clean Promotion Property.} If $\bin(a, j, x) \neq \bin(a, j + 1, x)$, then
$$\bin(a, j + 1, x) \in (m_{a, j}, m_{a, j + 1}].$$
In other words, whenever a partial expansion is performed, the only keys that move
are the keys that are assigned to the newly added bins.
\item \textbf{Independence. } For any given $ a, j $, the function $\bin(a, j, x)$
is mutually independent across all $ x \in U $.
\item \textbf{Near Uniformity. } For every $a, j, x$ and for every $\ell \in [m_{a, j}]$,
$$\Pr[\bin(a, j, x) = \ell] = (1 + O(1 / s)) \cdot \frac{1}{m_{a, j}}.$$ 
\end{itemize} 

Whereas independence and near uniformity are necessary for any addressing scheme (even in a fixed-size hash table), the clean promotion property is what glues together the outcomes of $\bin(a, j, x)$ for different values of $a$ and $j$. It ensures that only roughly a $1/s$-fraction of elements will have their address changed by any given partial expansion.

A consequence of the clean promotion property is that the functions $\bin(a, j, x)$ and $\bin(a, j + 1, x)$
must be closely related to one another. Thus a natural approach is to define the function $\bin(a, j, x)$ 
recursively, so that $\bin(a, j, x)$ depends on $\bin(a', j', x)$ for $a' \le a$ and $j' \le j$. 
In 1980, Larson gave an elegant construction \cite{larson1980linear} showing that such a recursive approach
 is indeed possible; Larson's scheme can be used to construct a solution $\bin(a, j, x)$ 
to the dynamic bin addressing problem
that can be evaluated in logarithmic expected time. Larson's scheme has found many applications to
external-memory problems, but the $\Omega(\log n)$ evaluation time has prevented it from being useful for 
internal memory hash tables.

This section shows that, somewhat remarkably, it is possible to achieve the clean promotion property
without recursion, and it is even possible to construct a solution to the dynamic bin addressing problem
that can be evaluated in $O(1)$ \emph{worst-case} time. We call our solution, which we present in 
Subsection \ref{subwaterfall}, 
\defn{waterfall addressing}. 

\paragraph{Efficiently finding which records need to be moved.}
So far we have focused on how to map records to a dynamically changing set of bins, but this alone does not fully solve the problem of how to dynamically 
resize a hash table. The clean promotion property ensures that each partial expansion moves only $O(n/s)$ records.
But how do we efficiently locate those records without performing a full scan
through the table?\footnote{Note that this is a problem that only arises for partial \emph{expansions} and not for the reverse operation which is a partial \emph{contraction}. In particular, when performing a partial contraction, the elements that need to be moved are precisely the ones that reside in the part of the table being eliminated.}

In Subsection \ref{subwhichrecords}, we show that it is possible to incorporate waterfall addressing into  
Iceberg hashing in a way that solves this problem. In particular, by adding $O(\log s)$ bits of 
overhead to each element in the hash table, we make it possible to locate in time $O(n / s)$ which records need 
to be moved. Our solution is not specific to Iceberg hashing; it can just 
as well be used with any hash table that
stores the majority of its elements in an array organized into bins.

\paragraph{Incorporating waterfall addressing into Iceberg hashing.}
Finally, when applying waterfall addressing to Iceberg hashing, there 
are several additional technical challenges that arise. These challenges are specific to
Iceberg hashing, and in particular, to how the probabilistic guarantees on the size of the 
backyard of the hash table interact with the dynamic resizing.
We show how to solve these issues in Subsection \ref{sec:resizing_iceberg}.
In doing so, we obtain a version of Iceberg hashing that is fully dynamic.

\subsection{Waterfall Addressing}\label{subwaterfall}
In this subsection, we describe a constant-time solution to the dynamic bin-addressing problem. 

To simplify discussion, we will think of the bins as being broken into \defn{chunks} of $E = 2^a / s$ bins during the $2^a$ doubling  (i.e., from $2^a$ bins to $2^{a+1}$ bins). 
At the beginning of this doubling there are $s$ chunks and at the end there are $2s$ chunks. Furthermore, when we refer to the \defn{size} of a chunk (or of the table as a whole), we shall be referring to the number of bins.

\subsubsection{A starting place: Larson's recursive scheme.}
The clean promotion property was introduced in 1980 by
Larson~\cite{larson1980linear}, who gave an elegant technique for achieving the property and applied the approach to cache-efficient linear hashing.

In Larson's scheme, each key $x$ has an (infinite) sequence of \defn{chunk hash
functions} that are used during the $2^a$ doubling (i.e., from $2^a$ bins to $2^{a+1}$ bins):
  $$g^{(a)}_1(x), g^{(a)}_2(x), \ldots, $$
where each $g^{(a)}_i:U \rightarrow [2s]$ maps elements uniformly to chunks. 
For each $r \in [s + 1, 2s]$, define $G^{(a)}(x,r)$ to be $g^{(a)}_i(x)$ 
for the smallest $i$ such that $g^{(a)}_i(x) \le r$. The definition of 
$G^{(a)}(x, r)$ satisfies two elegant properties: that (a) $G^{(a)}(x, r)$ is uniformly random in $[r]$; and that (b) either $G^{(a)}(x, r+1) = G^{(a)}(x, r)$ or $G^{(a)}(x, r) = r+1$. 

Larson's scheme computes 
$\bin(a, j, x)$ as follows. First, recursively compute $p = \bin(a, 0, x)$ to be the position that $x$ would reside
  in if the table had $2^a$ bins.  Then set 
$$
\bin(a, j, x) = 
\begin{cases}
p &\text{ if } G^{(a)}(x, j) \le s\\
G^{(a)}(x, j) \cdot E + \left(p \pmod E\right) & \text{ otherwise. } 
\end{cases}
$$
That is, if  $G^{(a)}(x, j)$ returns a non-expansion chunk, then we do not move the item; otherwise, we use the chunk hash function $G^{(a)}(x, j)$ to determine the high-order bits of the new address and the old address $p$ to determine the low-order bits.

The pseudocode for recursively computing the bin address for a record $x$ is given by Algorithm \ref{alg:Larson}.

  \begin{algorithm}
  \caption{Larson's Address Computation: Computing Bin$(a, j, x)$ \label{alg:Larson}}
  \begin{algorithmic}[1]
    \Require{Suppose we are doubling from size $2^a$ to $2^{a + 1}$ in
      chunks of size $E = 2^a / s$ bins. This function
      computes record $x$'s bin number after the $j$-th partial
      expansion.}

    \Statex

    \IIf{$2^{a + 1} = s$ and $j = s$} \Return{$g^{(a)}_1(x)$} \EndIIf \Comment{Base case}

    \Let{$p$}{Bin$(a - 1, s, x)$} \Comment{$p$ is the bin assignment for $x$ when the table was size exactly $2^a$}

    \Let{$i$}{$1$}

    \While{$g^{(a)}_i(x) > s + j$}
    \Let{$i$}{$i + 1$}
    \EndWhile

    \If{$g^{(a)}_i(x) \le s$}
    \Return{$p$}
    \Else{}
    \Return{$g^{(a)}_i(x) \cdot E + \left(p \pmod E\right)$}
    \EndIf
 \end{algorithmic}
\end{algorithm}

\begin{lem}[\cite{larson1980linear}]
  Larson's address computation function $\bin(a, j, x)$ satisfies the
  clean promotion property and maps each record $x$ uniformly at
  random to a bin in $[2^{a} + jE]$. 
\end{lem}

The recursive structure of Larson's scheme causes it to take time $\Omega(a)$ (and expected time $\Theta(a)$). The main contribution of this section is \defn{waterfall addressing}, a technique that improves this running time to $O(1)$ worst case. We begin by showing how to reduce the expected time to $O(1)$.

\subsubsection{Waterfall addresses in constant expected time.}
We modify Larson's scheme by introducing  a \defn{master hash function}  $m(x): U \rightarrow [2^w]$.
Whenever a key $x$ is moved into a new chunk by a partial expansion, we use $m(x)$ to determine the low-order bits of $x$'s address, rather than the recursively computed $p$. That is, we simply set the offset to be $m(x) \pmod
E$.  
Algorithm \ref{alg:waterfall} gives the pseudocode for this addressing
scheme (compared to Algorithm \ref{alg:Larson}, line 2 gets removed and lines 7 and 8 get modified). 

  \begin{algorithm}
  \caption{Waterfall Address Computation: Computing Bin$(a, j, x)$ \label{alg:waterfall}}
  \begin{algorithmic}[1]
    \Require{Suppose we are doubling from size $2^a$ to $2^{a + 1}$ in
      chunks of size $E = 2^a / s$ bins. Let $m(x)$ be the
      master hash function. This function computes record $x$'s bin
      number after the $j$-th partial expansion. The smallest allowable table size is $s$.}

    \Statex

    \IIf{$2^{a + 1} = s$ and $j = s$} \Return{$g^{(a)}_1(x)$} \EndIIf \Comment{Base case}

    \State{\xout{\color{red}$p  \leftarrow  \bin(a - 1, s, x)$}} \Comment{We no longer need to recursively compute $p$.}

    \Let{$i$}{$1$}

    \While{$g^{(a)}_i(x) > s + j$}
    \Let{$i$}{$i + 1$}
    \EndWhile

    \If{$g^{(a)}_i(x) \le s$} 
    \Return{\color{red} Bin$(a - 1, s, x)$}  \Comment{This is the only case where we recurse.}
    \Else{}
    \Return{\color{red} $g^{(a)}_i(x) \cdot E + \left(m(x) \pmod E\right)$} \Comment{Use the master hash instead of recursing.}
    \EndIf
 \end{algorithmic}
\end{algorithm}

\begin{figure}
    \centering
    \includegraphics[width=0.5\textwidth]{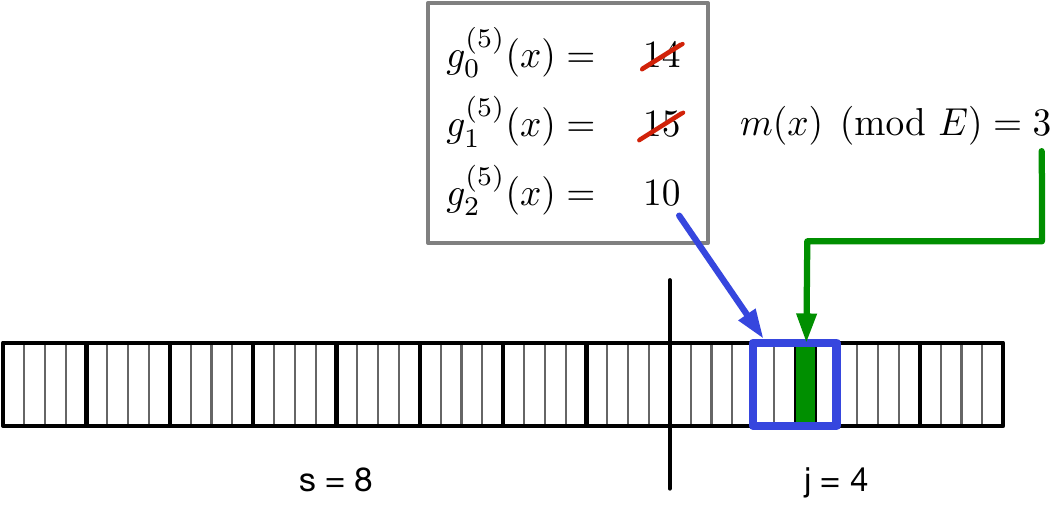}
    \caption{An example of the bin-selection algorithm in waterfall addressing as in Algorithm~\ref{alg:waterfall}. In this example $\textrm{Bin}(5, 4, x)$ is computed. The chunk that $x$ hashes to is determined by successively computing $g_0^{(5)}(x), g_1^{(5)}(x),\ldots$, until a value at most $s+j=12$ is found. This value determines the chunk that contains $x$, and then the bin within the chunk is determined by $m(x) \pmod{E}$.}\label{fig:waterfall}
\end{figure}

We call the addressing scheme \defn{waterfall addressing} because, as the table grows, more and more of the bits of $x$'s address are determined by $m(x)$.
Different records $x$ converge towards matching
$m(x)$ at different rates, together forming a sort of
``waterfall''. In a single table there will simultaneously be records
$x$ that agree with $m(x)$ in almost all of their bits, and (far
fewer) records that agree with $m(x)$ in only a few bits (these
records are at the ``top'' of the waterfall). An example of waterfall addressing is shown in Figure~\ref{fig:waterfall}.

When analyzing waterfall addressing, it will be useful to note that every recursively called subproblem of Algorithm \ref{alg:waterfall} has a power-of-$2$ number of bins (that is, $j = s$). In this case, the pseudocode for Algorithm \ref{alg:waterfall} simplifies considerably. Since we always use $g_1^{(a)}(x)$, rather than having to find some $g_i^{(a)}(x)$, we are able to skip lines 3--6 of Algorithm \ref{alg:waterfall}, resulting in Algorithm \ref{alg:waterfall2}.

\begin{algorithm}
  \caption{Computing Bin$(a, j, x)$ when $j = s$ \label{alg:waterfall2}}
  \begin{algorithmic}[1]
    \Require{We compute record $x$'s bin
      number after the $s$-th partial expansion. }

    \Statex

    \IIf{$2^{a + 1} = s$} \Return{$g^{(a)}_1(x)$} \EndIIf \Comment{Base case}

    \If{$g^{(a)}_1(x) \le s$}
    \Return{Bin$(a - 1, s, x)$}
    \Else{}
    \Return{$g^{(a)}_1(x) \cdot E + \left(m(x) \pmod E\right)$}
    \EndIf
 \end{algorithmic}
\end{algorithm}

We now give an analysis of (this basic version of) waterfall addressing.
\begin{lem}
  Waterfall addresses can be evaluated in $O(1)$ expected time (Algorithm \ref{alg:waterfall}). Moreover, for a given record $x$,
  Bin$(a, j, x)$ is uniformly distributed across $[2^a + jE]$.
\end{lem}
\begin{proof}
We begin by analyzing the running time. Lines 3--6 of Algorithm \ref{alg:waterfall} take constant expected time since each iteration of the while loop has at least a $1/2$ probability of terminating. In subsequent levels of recursion, the algorithm reduces to Algorithm \ref{alg:waterfall2}, which takes constant time per layer of recursion. Algorithm \ref{alg:waterfall2} has exactly a $1/2$ probability of terminating in each level of recursion (because $\Pr[g^{(a)}_1(x) \le s] = 1/2$). Thus the expected time to evaluate Algorithm \ref{alg:waterfall2} is also constant.

Next, we argue that bin assignments are performed uniformly at random. Suppose by induction that this is true for all $a'<a$. The value of $g_i^{(a)}(x)$ used in lines 6--7 of Algorithm \ref{alg:waterfall} is uniformly random in $[s + j]$. Thus, with probability $s/(s + j)$, $x$ is assigned to the first $s$  chunks using the recursively computed value of $\text{Bin}(a - 1, s, x)$, which we know by induction to be uniform in $[2^a]$. On the other hand, with probability $j/s$, $x$ is assigned to a random one of the final $j$ chunks and is then given a random offset into the chunk using the master hash. This process assigns $x$ uniformly at random in $[2^{a} + 1, 2^{a} + jE]$. Combining the two cases, $x$ is assigned to a random position in $[2^a + jE]$.
\end{proof}

\subsubsection{Waterfall addresses in worst-case constant time.}

For Iceberg hashing, we want worst-case constant-time operations. To
this end, we now define \defn{truncated waterfall addressing}.

One of the main issues with waterfall addressing (and Larson's scheme before it) is that we may need to evaluate an arbitrarily long sequence of chunk hash functions $g^{(a)}_i$.  
Truncated waterfall addressing truncates the sequence $\{g^{(a)}_i(x)\}$ to end at $i = \log s$. A priori, this does not necessarily seem like progress, since (a) it introduces an issue of what to do on a \defn{truncation overflow}, that is, when $g^{(a)}_i(x) > s + j$ for all of $i \in \{1, 2, \ldots, \log s\}$, so that the search for an address does not terminate within the first $\log s$ $g^{(a)}_i$s; and (b) it does not appear to get us any closer to a worst-case constant-time waterfall addressing scheme.  We will address these issues one after another, first showing how to fix truncated waterfall addressing in the case where there is no valid $g^{(a)}_i,$ and then showing how to compute truncated waterfall addresses in constant time.

\paragraph{What to do on truncation overflow.} When a truncation overflow occurs, we fall back to assigning $x$ to reside among the first $s$
chunks using the recursively computed Bin$(a - 1, s,
x)$. That is, if none of the $g^{(a)}_i(x)$'s are usable, then $x$ is assigned to the same position to which it would have been assigned at the end of the previous doubling, when there were exactly $2^a$ bins. Pseudocode is given in Algorithm \ref{alg:waterfalltruncated}; the changes from Algorithm \ref{alg:waterfall} are in red.

  \begin{algorithm}
  \caption{Truncated Waterfall Address Computation: Computing Bin$(a, j, x)$ \label{alg:waterfalltruncated}}
  \begin{algorithmic}[1]
    \Require{Suppose we are doubling from size $2^a$ to $2^{a + 1}$ in
      chunks of size $E = 2^a / s$ bins. Let $m(x)$ be the
      master hash function. This function computes record $x$'s bin
      number after the $j$-th partial expansion. The smallest allowable table size is $s$.}

    \Statex

    \IIf{$2^{a + 1} = s$ and $j = s$} \Return{$g^{(a)}_1(x)$} \EndIIf \Comment{Base case}

    \Let{$i$}{$1$}

    \While{$g^{(a)}_i(x) > s + j$ {\color{red} \text{ and } $i \le \log s$}}\Comment{Truncation condition in red}
    \Let{$i$}{$i + 1$}
    \EndWhile

    \If{$g^{(a)}_i(x) \le s$ {\color{red} \text{ or } $i > \log s$}}
    \Return{Bin$(a - 1, s, x)$}\Comment{Truncation condition in red}
    \Else{}
    \Return{$g^{(a)}_i(x) \cdot E + \left(m(x) \pmod E\right)$}
    \EndIf
 \end{algorithmic}
\end{algorithm}

The next lemma establishes that the bin assignments performed by truncated waterfall addressing are nearly uniform, as required in the dynamic bin addressing problem.

\begin{lem} \label{lem:nearlyuniform}
  If there are $k$ bins, then for each record $x$ and each bin $b$, the probability that truncated waterfall addressing maps $x$ to $b$ is $\frac{1}{k} \cdot (1 + O(1/s))$.
\end{lem}
\begin{proof}
  Suppose that $2^a < k \leq 2^{a+1}$ and let $j$ be the number of partial expansions that have occurred, that is, $k = 2^a + jE$. 

  The probability that $x$ experiences a truncation overflow is 
  $$
      \prod_{i = 1}^{\log s} \Pr[ g^{(a)}_i(x) > s + j] 
   \le \prod_{i = 1}^{\log s} \Pr[ g^{(a)}_i(x) > s] 
       = \frac{1}{2^{\log s}} 
       = 1/s.
  $$
    On the other hand, if a truncation overflow occurs then $x$ is assigned to bin Bin$(a-1, s, x)$. Importantly, Bin$(a-1, s, x)$ is uniformly random in the first $2^a$ bins, since truncated waterfall addressing and (non-truncated) waterfall addressing are equivalent in the case where the table size is an exact power of $2$ (in this case, both algorithms reduce to Algorithm \ref{alg:waterfall2}). 
  
  In summary, each key $x$ has only a $O(1/s)$ probability of being addressed differently by the two algorithms, and if $x$ is addressed differently, then truncated addressing assigns $x$ uniformly at random among $2^a = \Theta(k)$ bins. This implies the lemma.
\end{proof}

\paragraph{Truncated waterfall addressing in worst-case constant time. } 
Our final task is to compute truncated waterfall addressing in constant time.  This will require dealing with two issues: how to efficiently find the first $g^{(a)}_i(x)\leq s+j$ and how to eliminate the recursion.  As we shall see, the first problem can be dealt with by standard bit-manipulation techniques, whereas the second problem requires a more interesting algorithmic solution.

Notice that all of $g^{(a)}_1(x), g^{(a)}_2(x), \ldots, g^{(a)}_{\log s}(x)$ consume $O(\log^2 s)$ bits. As long as $s$ is not too large (i.e., $\log^2 s = O(w)$) it follows that the entire sequence $g^{(a)}_1(x), g^{(a)}_2(x), \ldots, g^{(a)}_{\log s}(x)$ can be packed into a single machine word (using the definition of packing given in Section \ref{sec:iceberg-hashing}), which we will denote by $\phi^{(a)}(x)$. Moreover, we can compute the entire sequence in constant time by computing $\phi^{(a)}(x)$ as a single $O(\log^2 s)$-bit hash of $x$, and then zeroing out every $\log s + 1$st bit in order to add appropriate padding.

By performing bit manipulation on $\phi^{(a)}(x)$, we can perform Lines 3--6 of Algorithm
\ref{alg:waterfalltruncated} in constant time.

\begin{lem} \label{lem:findgi}
  
  Let $r$ and $b$ be integers so that $r(b+1) \leq w$.  Let $\phi_1, \phi_2, \ldots, \phi_r$ be $b$-bit numbers packed into word $\phi$, and let $q$ be a $b$-bit number.  In constant time, one can determine the minimum $i$ such that $\phi_i \leq q$, or return $i=-1$ if no such $i$ exists.
 
\end{lem}

\begin{proof} 
The proof follows the same approach as Lemma \ref{lem:membership}. Pack $k$ copies of $q$ in a new word $Q$.  Compare $Q$ with $\phi$ and return the most significant $1$-bit of the resulting indicator word.
\end{proof}

Although lines $3$--$6$ of algorithm \ref{alg:waterfalltruncated} can be evaluated in constant time, there is still the issue of the recursion on line $7$ causing a potentially superconstant running time. Since the recursive subproblems are always on a power-of-two number of bins, the challenge becomes to evaluate Algorithm \ref{alg:waterfall2} in constant time.

Define the \defn{promotion sequence} $P(x)$ for $x$ to be the indicator word where
\[P_i(x) = \begin{cases} 1 & \text{ if } g^{(i)}_1(x) > s \\ 0 & \text{ otherwise}. \end{cases} \]
Another way to view Algorithm \ref{alg:waterfall2} is that we are finding the largest $a' \leq a$ such that $P_{a'}(x) = 1$, and we are then returning 
\[g^{(a')}_1(x) \cdot E' + \left(m(x) \pmod {E'}\right)\]
where $E' = 2^{a'} / s$. Thus, the task of computing Algorithm \ref{alg:waterfall2} reduces to the task of computing $a'$.

If we were given the promotion sequence $P(x)$, then we could determine $a'$ in constant time by standard
bit manipulation (see the discussion of bit manipulation in Section \ref{sec:iceberg-hashing}). The problem is that $P(x)$ consists of one bit from each of $\phi^{(1)}(x), \phi^{(2)}(x), \ldots$, each of which individually takes constant time to compute.

To fix this problem, we introduce one final algorithmic idea, reversing the relationship between $P(x)$ and $\phi^{(1)}(x), \phi^{(2)}(x), \ldots$. Let $P(x)$ be the output of a random hash function, and let $\psi^{(1)}(x), \psi^{(2)}(x), \ldots$ be hash functions each of which maps $x$ to a word that packs $\log s$ hashes of $\log s$ bits each (for a total of $\log s (1+ \log s)$ bits). Then we define $\phi^{(i)}(x)$ to equal $\psi^{(i)}(x)$ except with its most significant bit (i.e., the most significant bit of $g^{(i)}_1(x)$) overwritten by $P_i(x)$. That is, rather than using one bit from each $\phi^{(i)}(x)$ to determine $P(x)$, we use $P(x)$ to determine one bit in each $\phi^{(i)}(x)$. Importantly, the construction of $\phi^{(i)}(x)$ is overwriting the most significant bit of $\psi^{(i)}(x)$ with a \emph{random} bit, so $\phi^{(i)}(x)$ is still random. On the other hand, the construction makes it so that $P(x)$ is just a hash of $x$ and can be computed in constant time. Using $P(x)$, we can then determine $a'$ in constant time using standard bit manipulation, as desired. 

Putting the pieces together we arrive at the following theorem which establishes that truncated waterfall addressing is a constant time solution to the dynamic bin addressing problem.

\begin{thm}
Suppose that $\log s (1 + \log s) \le \word$ where $\word$ is the machine word size. Then, truncated waterfall addressing can be computed in constant time, satisfies the clean promotion property, and selects each of $k$ bins with probability $1/k \cdot (1 + O(1/s))$.
\end{thm}

\begin{proof}
  When evaluating Algorithm \ref{alg:waterfall2}, we can use Lemma \ref{lem:findgi} to perform the first level of recursion in constant time. The next level of recursion is guaranteed to be in the case where there are a power-of-$2$ bins. This case can be evaluated in constant time using the promotion sequence.
  
  The fact that truncated waterfall addressing satisfies the clean promotion property follows from the definition. The fact that bin assignment is nearly uniform follows from Lemma \ref{lem:nearlyuniform}.
\end{proof}

\subsection{Determining Which Records to Move}\label{subwhichrecords}
The clean promotion property ensures that, when a partial expansion occurs, the number of records whose address changes will be roughly a $1/s$ fraction of all records. 
If there are $n$ total records, then we wish to be able to identify which records to move in time $O(n/s)$. This means that we cannot simply traverse the table to find the records.

In this subsection, we show how to add a small amount of metadata to each bin so that we can efficiently detect which records to move during a partial expansion of truncated waterfall addressing. For simplicity, we will restrict ourselves to the case of $s \le \polylog n$ since it is the case that we will care about for Iceberg hashing.

\paragraph{Linked lists in each bin.} In this subsection, we will assume there are $O(n/s)$ bins and that the contents of each bin are stored contiguously in an array (note that this is not quite true for Iceberg hashing, since Iceberg hashing stores some elements in a backyard, but we will handle this issue later). Within each bin $b$, we maintain $s$ linked lists $L_1(b), L_2(b), \ldots L_s(b)$, where $L_\ell(b)$ consists of the records whose next address change will occur on the $\ell$th partial expansion of either the current doubling or some future doubling.

In more detail, for a record $x$ in bin $b$, the value $\ell$ can be computed as follows. Suppose we are currently doubling from $2^a$ to $2^{a+1}$ bins and that we have completed $j$ partial expansions, so there are $2^a + jE$ bins. If $g^{(a)}_1(x) > s + j$, then $x$'s current bin assignment must be determined by $g^{(a)}_i(x)$ for some $i > 1$. Then for all $q \in [1,i)$, we have $g_q^{(a)} > s+j$ and thus
\begin{equation}
    \ell = \min_{q \in [1,i)} g^{(a)}_q(x).
    \label{eq:l1}
\end{equation}
On the other hand, if $g^{(a)}_1(x) \le s + j$, then $x$ is in its final position for the current doubling. Suppose that $x$'s next promotion is during the $2^{a'}$ doubling, that is, 
$a' = \mathrm{argmin}_{a''>a} \{P_{a''}(x) = 1\}$.\footnote{If no such $a'$ exists, then we can feel free to not place $x$ in any linked list.}
Let $i = \mathrm{argmin}_{i'}\{g^{(a')}_{i'}(x) \le s\}$ (or $i = \log s + 1$ if no such $i$ exists).
In this case, 
\begin{equation}
 \ell = \min_{q \in [1,i)} g^{(a')}_q(x).
 \label{eq:aprime}
\end{equation}
We denote a record $x$'s choice of $\ell$ by $\ell^{(a)}_j(x)$, where $j$ is the number of partial expansions we have performed so far in the $2^{a}$ doubling.

\paragraph{Why the linked lists help.}
When we are performing the $j$th partial expansion, we need only examine the linked list $L_j(b)$ for each bin $b$. Not all of the elements of $L_j(b)$ will necessarily move during the partial expansion (some of them will move during future doublings), but all of the elements that we wish to move will be in a linked list $L_j(b)$ for some $b$. The next lemma bounds the total number of elements that are examined during a partial expansion. 

\begin{lem}
  Let $n$ be the number of records, let $k$ be the number of bins, suppose $s \leq \polylog n$, and let $j \in [s]$. Then w.s.h.p. in $n$,
  \begin{equation}
  \sum_{b = 1}^{k} |L_j(b)| = O(n/s).
  \label{eq:listsize}
  \end{equation}
  \label{lem:listsizes}
\end{lem}

\begin{proof}
  It suffices to bound the expected value of \eqref{eq:listsize}, since the lemma then follows by a Chernoff bound. Let $L = \bigcup_{b = 1}^{k} L_j(b)$. There are two cases for an  element $x \in L$: 
  \begin{itemize}
      \item Case 1: $x$'s address changes during the $j$th partial expansion, meaning that $x$ gets moved into the $(s + j)$th chunk. By Lemma \ref{lem:nearlyuniform}, the expected number of records $x$ in this case is $O(n/s)$.
      \item Case 2: $x$'s address does not change again a future $2^{a'}$ doubling, $a'>a$. That is, $a'$ is the smallest $a' > a$ such that $P_{a'}(x) = 1$ (or, equivalently, $g^{(a')}_1(x) > s$). In this case, the probability that $x \in L$ is at most 
      $$\Pr[\ell^{(a')}_{\ell - 1}(x) = \ell \mid g^{(a')}_1(x) > s].$$
      Since $\Pr[g^{(a')}_1(x) > s] = 1/2$, the above probability is at most
      $$2\Pr[\ell^{(a')}_{\ell - 1}(x) = \ell].$$
      On the other hand, by Lemma \ref{lem:nearlyuniform}, each record has a $O(1/s)$ chance that $\ell^{(a')}_{\ell - 1}(x) = \ell$. Thus, the expected number of records in this case is $O(n/s)$.
  \end{itemize}
\end{proof}

\paragraph{Maintaining the lists.} When maintaining the linked lists, there are two concerns: the space consumed by the linked lists, and the time needed to update the linked lists per hash table operation.

Because each linked list is confined to a single bin, it can be implemented using pointers consisting of $\Theta(\log h)$ bits. Indeed, because each bin has capacity $\Theta(h)$, pointers within the bin can be implemented as numbers between $1$ and $\Theta(h)$, thereby requiring only $\Theta(\log h)$ bits each. Assuming that $h \le \polylog n$, the linked lists introduce at most $O(\log h) = O(\log \log n)$ bits of space overhead per key, or $O(n \log \log n)$ bits of space overhead in total.

The larger issue is how to compute $\ell^{(a)}_j(x)$ in constant time for a given record $x$. Here, we make use of the following remarkable fact.

\begin{lem}
    Let $s_1, s_r, \ldots s_k$ be $b$-bit numbers packed into word $S$ so that the numbers and padding bits take no more that $\sqrt{\word}/3$ bits.  Then $\min_{i = 1}^{k} S_i$ can be computed in constant time.  
  
  \label{lem:min}
\end{lem}

\begin{proof}
The idea behind this proof is to construct two words $A$ and $D$, where $A$ contains $k$ copies of $s_1,\ldots,s_k$ and $D$ consists of $k$ copies of $s_1$ followed by $k$ copies of $s_2$ and so on. Then  $O(1)$ word operations on $A$ and $D$ can be used to compared every pair of $s_i$, $s_j$, yielding a comparison indicator word $E$ from which we compute the minimum.

Let $A$ be a word containing $k$ copies of $s_1, \ldots, s_k$, with a copy of the sequence appearing every $3(b+1)k$ bits.  Let $B$ be a word containing $k$ copies $s_1, \ldots, s_k$, with a copy of the sequence appearing every $3(b+1)(k-1)$ bits.  Mask out all but the numbers stored at multiples of $3(b+1)k$ from $B$ and call this $C$.  So $C$ has $s_1$ right justified in the last $3(b+1)k$ bits, $s_2$ in the preceding $3(b+1)k$ bits, etc.  Let $D$ consist of $k$ copies of $C$, with each copy shifted left by $b+1$ positions.  Now $D$ consists of $k$ copies of $s_1$, then $k$ copies of $S_2$, etc.  Compare $D$ with $A$ to perform an all pairwise comparison between the $s_i$ (recall that we discuss how to perform comparisons of packed machine words in Section \ref{sec:iceberg-hashing}).  

Let $E$ be the resulting comparison-indicator word.  Consecutive indicator bits are always separated by $b$ bits, and we set these $b$ bit separations to consist of all $1$s.  Now we are looking for a run of $(b+1)k$ 1s in a row, indicating that some $s_i$ is no greater than all the other $s_j$s.  We find this by adding one to the least significant position of each putative run, to see if the summation carries along the entire length of the run. We then identify the first such run by masking out all but the potential ``carry'' bits (one bit after each potential run) and computing the minimum $s_i$ from the position of the least significant carry bit.
\end{proof}

We now consider the task of computing $\ell^{(a)}_j(x)$. 
Assuming that we have $a'$ and $i$, then \eqref{eq:l1} and \eqref{eq:aprime} can be evaluated in constant time using Lemma \ref{lem:min}. The value of $a'$ in \eqref{eq:aprime} can be found in constant time by  using standard bit tricks on the promotion sequence. The value of $i$ (in either \eqref{eq:l1} or \eqref{eq:aprime}) can then be found using Lemma \ref{lem:findgi}. Thus we can obtain $\ell^{(a)}_j(x)$ in constant time. 

Putting the pieces together, we arrive at the following theorem.

\begin{thm}
  Let $n$ be the number of records. Assume $h = \Omega(s)$ and $h \le \polylog n$, where $\Theta(h)$ is the maximum bin size. Then, the linked lists $L_j(b)$ can be maintained in constant time per operation and induce at most $O(\log h)$ bits of overhead per key. Additionally, w.s.h.p.\ in $n$, the set of records that move during the next partial expansion can be identified in time $O(n/s)$. Finally, in the EM model using a cache line of some size $B = \Theta(h)$ and a cache of some size $M = \omega(B)$, the set of records that move during the next partial expansion can be identified with $O(n / B)$ cache misses.
    \label{thm:findrecords}
\end{thm}

It is worth taking a moment to better understand the $O(n/B)$ cache-miss bound. This bound follows simply from the fact that one can upper-bound the number of cache misses by the cost of performing a linear scan through the hash table (i.e., loading each bin into cache once). When we apply waterfall addressing to Iceberg hashing (Theorem \ref{thm:dynamic}), we will see that identifying the items that need to be moved is not the only source of cache misses (we must also actually move them; and we must also do some special-case handling for the backyard). However, the $O(n/B)$ cache misses spent identifying the items will continue to be the dominant term in the cost. Since $B = \Theta(h)$, and since Iceberg hashing performs rebuilds at most once every $O(n\log h / \sqrt{h})$ operations, the cache-miss cost of rebuilds per operation will end up being 
$$O\left(\frac{n/B}{n\log h / \sqrt{h}}\right) = O\left(\frac{\log h}{\sqrt{h}}\right).$$

We remark that partial contractions (that is, when a chunk is removed rather than added) are much simpler than partial expansions because the set of records that must be moved is readily apparent (they are the records in the chunk being removed).

\subsection{Implementing (Truncated) Waterfall Addressing in an Iceberg Hash Table}\label{sec:resizing_iceberg}

In this subsection, we describe how to implement waterfall addressing in
an Iceberg hash table in order to achieve efficient dynamic
resizing.

Because the backyard in an Iceberg hash table is so small, it can
be maintained using any (deamortized) resizing scheme. Thus our focus
will be on resizing the number of bins in the front yard of the
table. We use truncated waterfall addressing with partial expansions
(and contractions) to resize the table.

We use the standard Allocate Free Model of memory \cite{liu2020succinct}. If we are performing $s$
partial expansions per doubling, then the total number of memory
allocations for a table of size $n$ is $O(s \log n) \le \polylog n$. We will
assume that we have a large enough cache that pointers to the
allocated memory chunks can be cached at all times.

\paragraph{The main challenge: Maintaining the Iceberg analysis.}
Recall from the analysis of (static-size) Iceberg hashing that the Iceberg Lemma is used to upper bound the number of items in the backyard.  When we move items, if we are not careful, we may end up pushing extra items into the backyard and arriving at a state that cannot be analyzed by the Iceberg Lemma.  Thus our use of Waterfall addressing in Iceberg hashing, and its deamortization, ends up with some  complications in order to guarantee something fairly straightforward: that the state of the system (including who is in the front yard and who is the backyard) is consistent with an instantaneous expansion or contraction that  can be analyzed by the Iceberg Lemma.

More specifically, the issue that we must be careful about is the following.
Whenever we
move a record $r$ into a new bin $b$ during an expansion or contraction,
one can think of that move as representing a new insertion into the
bin $b$. But the timing of the insertion will be dependent on the bin
number $b$ (and on where the record was before the move), which means we cannot simply analyze the insertion as
being into a random bin. That is, we must analyze records $b$ that are
moving around due to a partial expansion or contraction differently than
we would treat records that are being inserted by the user.

\paragraph{Implementing partial expansions.} Suppose we are adding a
new chunk $C$, and let $t_0$ be the time at which we begin the
partial expansion.

At time $t_0$, we allocate memory to the chunk $C$. For each record
$x$, let $\g_{\text{old}}(x)$ denote the bin that $x$ would be assigned
to without chunk $C$ and let $\g_{\text{new}}(x)$ denote the bin that
$x$ would be assigned to with chunk $C$ present (i.e. after the partial expansion). For now let us assume
that, until the partial expansion is complete, queries will treat the
chunk $C$ as being \defn{semi-present}, meaning that a query for a
record $x$ will check both $\g_{\text{old}}(x)$ and
$\g_{\text{new}}(x)$. Of course, most records $x$ will satisfy
$\g_{\text{old}}(x) = \g_{\text{new}}(x)$, in which case the query is
unaffected.

Once $C$ has been allocated, the partial expansion is performed in
three parts.
\begin{itemize}
\item \textbf{The Preprocessing Phase:} In this phase, we construct a
  new counter in each bin that we call the \defn{demand counter}. The
  demand counter in bin $b$ keeps track of how many records $x$ in the
  table (including in the backyard and in other bins) satisfy
  either $\g_{\text{old}}(x) = b$ or $\g_{\text{new}}(x) = b$. (Importantly, this means that a
  single record could contribute to two different demand counters.)

  Later we will describe how to deamortize the phase. As
  the phase is performed, any concurrent operations also update the
  demand counters of the bins that they modify.

\item \textbf{The Time Freeze:} Let $t_1$ be the moment in time immediately after the Preprocessing Phase completes. We refer to $t_1$ as the time freeze point. Roughly speaking, we will try to simulate the partial
  expansion as having occurred instantaneously at time $t_1$.

  At time $t_1$, every bin reserves some of its slots for records that
  are currently in the table.\footnote{These reservations are performed logically at $t_1$ but do not require any physical action at $t_1$.} If a bin has demand counter $d$, and the
  capacity of the bin is $r = h + \tau_h$, then the bin reserves
  $\min(d, r)$ slots for records currently in the table. Any records
  that are currently in the bin are immediately given reserved slots.

\item \textbf{The Reshuffling Phase:} Call a record
  \defn{grandfathered} if it was in the table at time $t_1$ and has
  remained in the table since. The reshuffling phase identifies which records
  in the table (including both in the first and backyards) 
  are grandfathered\footnote{Since
    space efficiency in the backyard is not important, we can simply
    have separate tables for the grandfathered and non-grandfathered
    records.}, and attempts to move each grandfathered record $x$ to a
  reserved slot in $\g_{\text{new}}(x)$. If there is a free reserved
  slot in $\g_{\text{new}}(x)$, then $x$ is given that slot, and
  otherwise $x$ is sent (possibly back) to the backyard. If $x$ is
  being moved from $\g_{\text{old}}(x)$ in which it was taking up
  a reserved slot, then the number of reserved slots in that bin is
  decremented by $1$ (because $x$ is no longer present in that bin).
 
  Later we will describe how to deamortize the phase. During the phase, concurrent operations may take
  place. If a grandfathered record $x$ is deleted,
  then for each of the bins
  $b \in \{\g_{\text{old}}(x), \g_{\text{new}}(x)\}$, if $x$ was either
  residing in bin $b$ or if there is a free reserved slot in bin $b$
  (think of this slot as being reserved for $x$), then the operation
  that removes $x$ also decrements the number of reserved slots in bin
  $b$. If a new record $x$ is inserted during the phase (note that $x$
  is therefore \emph{not} grandfathered), and the only free slots in
  the $\g_{\text{new}}(x)$ are reserved, then $x$ is sent to the
  backyard despite there being free slots in the 
  $\g_{\text{new}}(x)$.

  Once the Reshuffling Phase is complete, the partial expansion is
  also complete. Call this time $t_2$.
\end{itemize}

One minor technical issue that we must be careful about during the Reshuffling Phase
 is that Iceberg hashing requires that no two records in a given bin have the same fingerprint. Thus, when placing a grandfathered record into a free reserved slot in a bin,
 we must handle the following additional two cases: if there 
 is another grandfathered record in the bin with the same
  fingerprint as $x$, then $x$ is sent to the backyard and the
  number of reserved slots in the bin is decremented by $1$ (i.e., the
  reserved slot given to $x$ is removed); if there is another
  \emph{non-grandfathered} record $y$ in the bin such that $y$ has the
  same fingerprint as $x$, then $y$ is sent to the backyard and
  $x$ is given the reserved slot.

Recall that each bin must keep a floating counter that tracks the number
of items that hash to the bin but are in the backyard. During the Prepossessing and
Reshuffling phases, we must be careful to keep the floating counters in
consistent states, as follows. During the Preprocessing Phase, the floating counters for
each bin $b \in C$ are initialized to be the number of records $x$ in
the backyard such that $\g_{\text{new}}(x) = b$. Then, during the
Reshuffling Phase, whenever the traversal visits a grandfathered
record $x$ in the backyard such that $\g_{\text{new}}(x) \in C$,
the floating counter for $\g_{\text{old}}(x)$ is decremented (in
essence, $\g_{\text{new}}(x)$ is now declared to be
responsible for record $x$, even if $x$ remains in the backyard).

\paragraph{Putting the pieces together.}
As described above, the purpose of the three phases in each partial expansion is to simulate
the expansion as having occurred at a single point in time $t_1$. In Appendix \ref{app:expansion},
we prove that partial expansions (implemented in this way) do not interfere with any of the properties
of Iceberg hashing (i.e., the backyard remains small, and elements individually have good probability of being in the front yard). The appendix also describes how to carefully implement the partial expansion
such that it is deamortized and I/O efficient, and describes how to analogously handle partial contractions.
The result is the following theorem: 

  \begin{thm}
    Consider a dynamic Iceberg hash table with average-bin-fill parameter $h$. 
    Suppose that the number $n$ of elements stays in the
    range such that $\log n / \log \log n \ge \Omega(h)$. Suppose that the
    table used in the backyard supports
    constant-time operations (w.h.p.\ in $n$), has load factor
    at least $1 / \poly(h)$, and is stable. Finally,
    set the resize granularity $s = \sqrt{h}$.
    
    Consider an operation on a key $x$. The following guarantees hold.
    \begin{itemize}
    \item \textbf{Time Efficiency.} The
      operation on $x$ takes constant time in the RAM model, w.h.p.\ in $n$.
    \item \textbf{Cache Efficiency.} Consider the EM model using a cache line of size $B \geq 2h$ and a cache of size $M \ge ch^{1.5}B + s \log n$ for some sufficiently
      large constant $c$, and suppose that each bin in the front yard is memory aligned, that is, each bin is stored in a single cache line.  Finally, suppose that the description bits of the hash functions are cached.
    Then the expected number of cache misses
      incurred by the operation on $x$ is $1 + O(1 / \sqrt{B}) = 1+o(1)$.
    \item \textbf{Space Efficiency.} The total space in machine words consumed by the table, w.h.p.\ in $n$, is
      $$\left(1 + O\left(\frac{\sqrt{\log h}}{\sqrt{h}}\right)\right) n = (1+o(1))n.$$ 
    \item \textbf{Stability.} 
    If a partial resize has not been triggered in the past $O(n/s)$ operations, then the table is stable.
    \end{itemize}
  \label{thm:dynamic}
\end{thm}

  \begin{rmk}
    In the case where the cache line size is $B = \Theta(h)$, the
    guarantees in Theorem \ref{thm:dynamic} come close to matching the
    best known bounds for external-memory hashing
    \cite{jensen2008optimality}. In particular,
    \cite{jensen2008optimality} achieves load factor
    $1 - O(1 / \sqrt{B})$ with an average of $1 + O(1/\sqrt{B})$ cache misses per
    operation.
  \end{rmk}

  \begin{rmk}
    Theorem \ref{thm:dynamic} requires
    $h \le O(\log n / \log \log n)$. Note, however, that whenever
    $\log n$ changes by more than a factor of two, we can simply
    rebuild our table (with a new parameter $h$ of our choice). These
    rebuilds can be performed space efficiently and are rare enough that they
    do not hurt the expected cache behavior of operations. In this
    sense, the assumption that $\log n / \log \log n \ge \Omega(h)$ is
    without loss of generality.
    
    In more detail, the rebuilds can be implemented as follows. Break
    the table's lifetime into \defn{doubling windows}, consisting of
    time windows in which the table's size either doubles or halves.
    Then place the doubling windows into \defn{window runs}, where each
    window run is determined as follows: if at the beginning of the
    window run the table size is $n$, then the window run lasts for a
    random number $k \in [1, \log n / 2]$ of doubling windows, after
    which the next window run begins. During the final window of each
    window run, we rebuild the hash table from scratch using the then appropriate
    value of $h$.
    
    Each rebuild can be performed space efficiently by storing both the new
    and old versions of the hash table as dynamically resized Iceberg
    hash tables during the rebuild. During a given rebuild, operations
    may incur multiple cache misses, but the probability of a given
    operation being contained in a window where a rebuild occurs is at
    most $O(1 / \log n)$ (where $n$ is the current table size), so the
    expected number of cache misses per operation remains
    $1 + O(1 / \sqrt{h})$.
    
    \label{rem:window-run-rebuilds}
  \end{rmk}

\begin{rmk}
Note that partial rebuilds occur at most once every $\Omega(n/s)$ operations. Thus if a hash table is changing size rapidly, it may forego stability for some (arbitrarily small) constant fraction of its operation. This is fundamental since, if a hash table is shrinking rapidly, then every time that its size halves (i.e., every $O(n)$ operations), the hash table must rearrange the remaining elements to occupy less total memory (at least, if the hash table wishes to be space efficient).\footnote{\emph{A priori}, chained hashing might seem to be an exception to this rule, since many implementations do achieve full stability. However, this is only achieved by assigning each element its own dynamically allocated portion of memory, and then using a $\Omega(\log n)$-bit pointer in order to reference that element. Even if we allow for memory allocations at arbitrary granularity, the pointer overheads preclude space efficiency in any fully stable hash table.} If, furthermore, the hash table wishes to incur $O(1)$ worst-case time per operation, then it follows that the hash table must be unstable for a constant fraction of its operations. On the other hand, in time periods where a hash table's size stays within a narrow band, Theorem \ref{thm:dynamic} guarantees full stability.

Of course, many hash tables take the simpler approach of implementing each partial resize to occur during a single operation (which, for Iceberg hashing, would take $O(n/s)$ time). This would mean that operations that \emph{do not} trigger a rebuild are $O(1)$ time w.h.p., and that operations that do trigger a rebuild are amortized $O(1)$ time. An advantage of this approach is that stability is easier to think about: every operation is stable, unless it triggers a rebuild. Moreover, this approach lends itself to simple locking schemes which have been shown to improve concurrency \cite{PandeyBeJo123}.
\end{rmk}

 \section{Reducing the Wasted Bits Per Key to $\mathbf{O(\log \log n)}$}\label{sec:space}

In this section we consider the problem of further optimizing the
space efficiency of Iceberg hashing.  We achieve a load factor of $1-O(\log\log n/\log n)$.  This improves on the previous best known bound~\cite{liu2020succinct} of $1-O(1/\sqrt{\log n})$.  In the case where keys and values have $\Theta(\log n)$ bits, our
table wastes only $O(\log \log n)$ bits per key (in comparison with
the previous state of the art of $O(\sqrt{\log n})$ bits per key).

So far, we have been limited by the
fact that the routing table in each bin can only support
$O(\log n / \log \log n)$ records. This, in turn, has limited the average-bin-fill parameter $h$ to
$O(\log n / \log \log n)$ and has limited our best achievable
load factor to $1 -  O(\sqrt{\log \log n}/\sqrt{\log n})$.

\paragraph{Supporting large bin sizes.}
Throughout the rest of the section, we consider average-bin-fill parameters $h$ such that
$$h \in \left[\frac{\log n}{\log \log n},\frac{\log^2 n}{\log \log n}\right].$$
Rather than having a single routing table per bin, we now have
$$k = \left\lceil \frac{h}{\log n / \log \log n} \right\rceil$$
routing tables $R_1, R_2, \ldots, R_k$ per bin, each of which can
route up to $2 \log n / \log \log n$ fingerprints.

Each key $x$ selects a routing table using a new hash function
$\rr: U \rightarrow [k]$ where $U$ is the universe of keys. Operations
on key $x$ use routing table $R_{\rr(x)}$ in $\g(x)$.

Although the key $x$ hashes to a specific routing table $R_{\rr(x)}$,
the key can still be placed anywhere within the $\g(x)$. That is, the
routing table $R_{\rr(x)}$ maps fingerprints to arbitrary positions in
$[h + \tau_h]$.

The assignment of keys to routing tables introduces a problem: some
routing tables $R_i$ may be assigned more than
$2 \log n / \log \log n$ keys to route. When this happens, the
routing table sends overflow keys to the backyard of the Iceberg
hash table. Keys sent to the backyard by an overflowed routing
table are called \defn{routing floaters}.

We once again use the Iceberg Lemma, which tells us that, w.s.h.p., there are very
few total routing floaters.
\begin{lem}
  Let $h\in[\log n / \log \log n , \polylog (n)]$, and
  define $N$ as in Section \ref{sec:iceberg-hashing}.  There are $O(N / \polylog(n))$ routing
  floaters in the table, w.s.h.p.\ in $N$.  Moreover, for a given key $x$, the
  probability that there is a routing floater $y$ such that
  $(\g(y), \rr(y)) = (\g(x), \rr(x))$ is at most $O(1 / \poly(h))$.
  \label{lem:routing_floaters}
\end{lem}
\begin{proof}
  The proof follows exactly as for Lemma \ref{lem:capacity_floaters},
  except that now the ``bins'' in the Iceberg Lemma are the routing
  tables rather than the actual bins in the Iceberg hash table. \end{proof}
  
  Note
  that the maximum capacity per routing table of
  $2 \log n / \log \log n$ is much larger than necessary for the
  analysis, since a capacity of
  $\log n / \log \log n + \tau_{\log n / \log \log n}$ would suffice for the proof of Lemma~\ref{lem:routing_floaters}.  We are able to apply this much slack to the routing tables because they are a low-order term in the space consumption of the Iceberg hash table.

\paragraph{Changes to the metadata.}
To accommodate the large value of $h$, the bookkeeping in each bin
also changes slightly. Each routing table maintains its own floating
counter, and queries on a record $x$ need only go to the backyard if the floating counter for $R_{\rr(x)}$ in $\g(x)$ has a
non-zero floating counter. Additionally, since $h$ may be much larger
than $\log n$, we can no longer keep track of the free slots in the
bin with a bitmap. Thus the vacancy bitmap is replaced with a free
list, which is a linked list of the free slots in the bin.

\paragraph{The non-resizing case.} We can now extend Theorem \ref{thm:static} to hold for $h \le \log^2 N / \log \log N$.

\begin{thm}
Consider an Iceberg hash table that never contains
more than $N$ elements and suppose that the average-bin-fill parameter satisfies  $h = O(\log^2 N / \log \log N)$. Suppose that the
    backyard table supports
    constant-time operations (w.h.p.\ in $n$), supports load factor
    at least $1 / \poly(h)$, and is stable. 

  Consider a sequence of operations in which the number of records in
  the table never exceeds $N$, and consider a query, insert, or delete
  that is performed on some key $x$. Then, the
  following guarantees hold.
  \begin{itemize}
  \item \textbf{Time Efficiency.} The
    operation on $x$ takes constant time in the RAM model, w.h.p.\  in $N$.
  \item \textbf{Cache Efficiency.} Consider the EM model using a cache line of size $B \geq 2h$ and a cache of size $M = \Omega(B)$, and suppose that each bin in the front yard is memory aligned, that is, each bin is stored in a single cache line. Finally, suppose that the description bits of the hash functions are cached.
    Then the operation on $x$ has probability at least $1 - 1 / \poly(B)$ of
    incurring only a single cache miss.
  \item \textbf{Space Efficiency.} The total
    space in machine words consumed by the table, w.h.p.\ in $N$, is
    $$\left(1 + O\left(\frac{\sqrt{\log h}}{\sqrt{h}}\right)\right) N = (1+o(1)) N.$$ 
  \item \textbf{Stability.} The hash table is stable.
  \end{itemize}
  \label{thm:static_succinct}
\end{thm}

\begin{proof}
  The proof is the same as for Theorem \ref{thm:static}, except with
  two changes for the case of
  $h \in [\log N / \log \log N, O(\log^2 N / \log \log N)]$.

  First, we must account for the space consumed by the routing tables
  $R_1, \ldots, R_k$ in each bin. Fortunately, these tables consume
  only $O(h \log h)$ bits per bin, in comparison to the
  $\Theta(h \log N)$ bits otherwise needed for the bin. Thus the
  routing tables only increase the total space consumption by a factor
  of at most
  $$1 + O\left(\frac{\log h}{\log N}\right) \le 1 + O\left(\frac{\log \log N}{\log N}\right) \le 1 + O\left(\frac{\sqrt{\log h}}{\sqrt{h}}\right),$$
  where the inequalities use that $h \in[\log N / \log \log N,O(\log^2 N / \log \log N)]$.

  Second, we must account for the presence of routing floaters in the
  backyard. This is handled by simply applying Lemma
  \ref{lem:routing_floaters}.
\end{proof}

\begin{cor}
In Theorem~\ref{thm:static_succinct}, when $h=\log^2 N/\log\log N$, the total space in machine words consumed becomes 
$$\left(1+ O\left(\frac{\log\log N}{\log N}\right)\right) N.$$
\label{cor:space_efficient_static}
\end{cor}

\paragraph{Supporting dynamic resizing.}
One can support dynamic resizing using essentially the same approach
as in Theorem \ref{thm:dynamic}.  The process of performing a partial
expansion or contraction must be slightly modified to accommodate the
routing-table structure of each bin, however. In particular, we now
maintain demand counters $d_{b, i}$ for each routing table $R_i$ in
each bin $b$. When a time freeze occurs, the bin reserves
\begin{equation}
  \min\left(\left(\sum_{i = 1}^k \min(d_{b, i}, 2 \log n / \log \log n)\right), h + \tau_h \right)
  \label{eq:reserve_amount}
\end{equation}
slots for records currently in the bin. That is, the bin reserves
$d_{b, i}$ slots per routing table, subject to the capacity
constraints of the routing tables and the bin. If
\eqref{eq:reserve_amount} is $h + \tau_h$, then the bin can determine
arbitrarily how many slots are reserved for each routing table, as
long as the $i$-th routing table has at most
$\min(d_{b, i}, 2 \log n / \log \log n)$ slots reserved and the total
number of reserved slots is $h + \tau_h$.\footnote{To simplify
  accounting, the number of slots reserved for each routing table can
  be determined lazily during the Reshuffling Phase. That is, only
  when the routing table $R_i$ is next accessed, do we decide how many
  slots were reserved for it at the time freeze.}

The proofs of lemmas analogous to Lemma
\ref{lem:dynamic_correctness1}, Lemma \ref{lem:dynamic_correctness2},
and Lemma \ref{lem:second_level_current_size} follow exactly as in
Section \ref{sec:resizing_iceberg} and Appendix \ref{app:expansion}, except that now routing floaters
are accounted for in addition to capacity floaters and fingerprint
floaters.\footnote{Since there are now multiple demand counters per bin, we must be
careful that the time (in the RAM model) to perform a partial
expansion or contraction is still $O(n / s)$, where $s$ is the resize 
granularity used by waterfall addressing. Fortunately, the current values for the
demand counter of each routing table, and the value of
$\sum_{i = 1}^k \min(d_{b, i}, 2 \log n / \log \log n)$ for each bin
$b$, are straightforward to keep track of at all times (rather than
just during the Preprocessing Phase) while adding only $O(1)$ overhead
per operation. Thus, in the case of a partial expansion, the
Preprocessing Phase needs only to instantiate these values in the new
bins, and in the case of a partial contraction, the Preprocessing Phase
needs only update the values appropriately to take account of the
$O(n / s)$ records that are being relocated.}

Putting the pieces together, we can extend Theorem \ref{thm:dynamic} to support larger values of $h$. 

  \begin{thm}
    Consider a dynamic Iceberg hash table with average-bin-fill parameter $h$. 
    Suppose that the number $n$ of elements stays in the
    range such that $h  = O(\log^2 n / \log \log n)$. Suppose that the
    table used in the backyard supports
    constant-time operations (w.h.p.\ in $n$), has load factor
    at least $1 / \poly(h)$, and is stable. Finally,
    set the resize granularity $s = \sqrt{h}$.
    
    Consider an operation on a key $x$. The following guarantees hold.
    \begin{itemize}
    \item \textbf{Time Efficiency.} The
      operation on $x$ takes constant time in the RAM model, w.h.p.\ in $n$.
    \item \textbf{Cache Efficiency.} Consider the EM model using a cache line of size $B \geq 2h$ and a cache of size $M \ge ch^{1.5}B + s \log n$ for some sufficiently
      large constant $c$, and suppose that each bin in the front yard is memory aligned, that is, each bin is stored in a single cache line.
    Then the expected number of cache misses
      incurred by the operation on $x$ is $1 + O(1 / \sqrt{B}) = 1+o(1)$.
    \item \textbf{Space Efficiency.} The total space in machine words consumed by the table, w.h.p\ in $n$,  is
      $$\left(1 + O\left(\frac{\sqrt{\log h}}{\sqrt{h}}\right)\right) n = (1+o(1))n.$$

    \item \textbf{Stability.} 
    If a partial resize has not been triggered in the past $O(n/s)$ operations, then the table is stable.

    \end{itemize}
  \label{thm:dynamic_succinct}
\end{thm}

\begin{cor}
In Theorem~\ref{thm:dynamic_succinct}, when $h=\log^2 N/\log\log N$, the total space in machine words consumed becomes 
$$\left(1+ O\left(\frac{\log\log N}{\log N}\right)\right) N.$$
\label{cor:space_efficient_dynamic}
\end{cor}

 \section{Achieving Subpolynomial Failure Probabilities}\label{sec:prob}

In this section, we consider the problem of achieving subpolynomial probabilities of failure
for Iceberg hashing assuming access to fully random hash functions (as in past work, \cite{goodrich2011fully, goodrich2012cache}, the assumption of fully random hash functions is needed to avoid failure probability that is introduced by the hash functions themselves).

We begin by stating a version of Theorem \ref{thm:dynamic_succinct} assuming fully random hash functions. The theorem
follows immediately from the w.s.h.p. guarantees offered by the lemmas in the previous sections. 
\begin{thm}[\Cref{thm:dynamic_succinct} with super-high probability]
In the conditions of \Cref{thm:dynamic_succinct}, suppose that the backyard table $\calT$ supports each operation in constant time with probability $1 - p(n)$. Then, assuming fully random hash functions, the guarantees of the Iceberg hash table hold with probability $1 - O(p(n) + 2^{-n/\polylog(n)})$ per operation.
\label{thm:dynamic_prob}
\end{thm}

We now consider the problem of designing a backyard hash table $\mathcal{T}$ that has a super small failure probability $p$ (the same failure probability can then be achieved by Iceberg hashing, using Theorem \ref{thm:dynamic_prob}). By employing the very-high probability hash table of Goodrich, Hirschberg, Mitzenmacher, and Thaler \cite{goodrich2012cache} one can achieve $p = 2^{-\polylog(n)}$.  In this section, we show how to do significantly better when each key is $\Theta(\log n)$ bits.
For this case, we are able to achieve
$p = O(2^{-n^{1-\epsilon}})$ for a positive constant $\epsilon$ of our choice. 

Throughout the rest of this section, set $\delta = \epsilon/4$, so we are aiming for $p = O(2^{-n^{1-4\delta}})$, and set the machine word size $w = \Theta(\log n)$. 

\paragraph{The difficulty of subpolynomial guarantees: not enough random bits.}
The main difficulty that one encounters when trying to achieve a failure probability
$p$ that is subpolynomial is that hash collisions must be treated as the common case.
That is, since any two keys have a $1/\poly(n)$ chance of colliding (on any $w = \Theta(\log n)$-bit hash function), we must be able to handle a superconstant number of keys
colliding on their hash functions. If we want $p = O(2^{-n^{1- 4 \delta}})$ then we must
be willing to tolerate $\Omega(n^{1-4\delta} / \log n)$ keys colliding with one another.

\paragraph{Storing $n^{1 - 2\delta}$ keys deterministically.}
In order to store a small set of $n^{1 - 2\delta}$ keys deterministically, we will make use of a radix trie with fanout $n^{\delta}$. We formalize the properties that we will need from the radix trie in the following lemma. 

\begin{lem}
Suppose keys are $w = \Theta(\log n)$ bits and let $\delta > 0$ be a constant.
There exists a deterministic data structure that can be initialized in time $o(n)$, that consumes space $o(n)$, and
that supports insertions, deletions, and queries in constant time on a 
set of up to $O(n^{1 - 2 \delta})$ keys.
\label{lem:trie}
\end{lem} 

\begin{proof}
As noted above, the data structure is a radix trie with fanout $n^{\delta}$. The root node $r$ of
the trie is an array of length $n^{\delta}$. The $i$th entry in $r$ 
is null if there are no keys $x$ whose first $\delta \log n$ bits equal $i$.
Otherwise, the $i$th entry points to a recursively-defined trie storing the 
final $w -  \delta \log n$ bits of each
key $x$ whose first $\delta \log n$ bits equals $i$.

The trie has depth $O(1/\delta) = O(1)$. Since the data structure stores $O(n^{1 - 2\delta})$
keys, the trie can have at most $O(n^{1 - 2\delta})$ nodes.
The total space consumption is therefore $O(n^{1 - \delta})$ since each node
consumes $n^{\delta}$ space.

The data structure requires $O(n^{1 - \delta})$ time to initialize, where the initialization time is spent
allocating $O(n^{1 - 2\delta})$ arrays that each consist of $n^{\delta}$ null
pointers. These arrays can then be used to implement operations on the trie in constant
time.
\end{proof}

\paragraph{Storing all but $O(n^{1 - 2\delta})$ keys in bins.}
We now describe a hash table with failure probability $p = O(2^{-n^{1 - 4\delta}})$. Because we are constructing a hash table to be used as a backyard, and thus we are
not concerned about space efficiency (a load factor of $\Theta(1)$ is okay), we can ignore
the issue of dynamic resizing (which can be performed with deamortized rebuilds) and
the issue of deletions (which can be performed by marking elements as deleted and then
rebuilding the data structure every $O(n)$ operations). Thus, we can assume there are 
$\Theta(n)$ records and that the only operations are queries and insertions.

We maintain $n / \log n$ bins, each with capacity $\Theta(\log n)$. Queries and insertions
are implemented in each bin using the dynamic fusion tree 
of P\v{a}tra\c{s}cu and Thorup \cite{patrascu2014dynamic},
which supports constant time deterministic operations on a set of size $\polylog n$.
If a bin overflows (that is, there are more than $c\log n$ records for some large constant $c$)
then the overflow records are stored in the data structure from Lemma \ref{lem:trie}.
Call these records \defn{stragglers}.

\begin{lem}
With probability $1 - O(2^{-n^{1 - 4\delta}})$ there are $O(n^{1 - 2\delta})$ 
stragglers at any given moment.
\label{lem:stragglers}
\end{lem}
\begin{proof}
The fact that we need only consider insertions allows for the following analysis.
The expected number of stragglers is $o(1)$ since each bin has a $1/\poly(n)$ probability
of overflowing. On the other hand, the number of stragglers is a function of $O(n)$ independent
random variables (i.e., the bin choice for each ball that is present), and each 
of these random variables can only affect the number of stragglers by $\pm 1$. Thus
we can apply McDiarmid's inequality (see Theorem \ref{thm:mc}) to obtain a concentration bound on
the number of stragglers. This implies that there are $O(n^{1 - 2\delta})$ 
stragglers with probability at least $1 - O(2^{-n^{1 - 4\delta}})$.
\end{proof}

Putting the pieces together, and using $\delta = \epsilon / 4$, we arrive at the following theorem.

\begin{thm}
Consider keys that are $\Theta(\log n)$ bits and let $\epsilon > 0$ be a constant. 
There is a hash table (using fully random hash functions) that supports constant time operations and constant load factor with failure
probability $1 - O(2^{-n^{1 - \epsilon}})$ per operation.
\end{thm}
\begin{proof}
This follows from Lemmas \ref{lem:trie} and \ref{lem:stragglers}.
\end{proof}

\begin{cor}
Consider keys that are $\Theta(\log n)$ bits and let $\epsilon > 0$ be a constant.
Iceberg hashing with fully random hash functions
can be implemented with failure probability $1 - O(2^{-n^{1 - \epsilon}})$ per operation.
\end{cor}
 \section{Succinctness Through Quotienting}\label{sec:quotient}

So far, we have focused on designing an explicit data structure, that is, a space-efficient data structure that explicitly stores each key-value pair somewhere in memory. Such a data structure does not achieve the information-theoretic optimum memory consumption, however. Given a set of $n$ keys from a universe $U$, the minimum number of bits needed to encode the set is
$$\log \binom{|U|}{n},$$
which by Stirling’s approximation is $n \log \frac{|U|}{n} - O(n)$. In this section, we give a succinct version of the dynamic Iceberg hash table that,  assuming that $|U| \le \poly n$, stores $n$ keys using space
$$n \log \frac{|U|}{n} + O(n \log \log n)$$
bits. The table can also support $v$-bit values for each key using an additional $v$ bits of space per key.

\paragraph{Using quotients to save space.}  Although we will remove the assumption later, for now let us assume that our keys are selected at random from the universe $U$. This means that the master hash $m(x)$ of each key can simply use the low-order bits of the key $x$. These bits, in turn, do not need to be explicitly stored in the hash table. (This space-saving technique is often called \defn{quotienting}).

\paragraph{Storing some keys with fewer bits than others.}  Recall that only part of each key's address is determined by its master hash, and that, at any given moment, different keys may use different numbers of bits from their master hash. All of the keys within a given bin use the same number of bits from their master hashes, however. In particular, the keys in bins whose indices are in the range $I_i = (s2^{i-1}, s2^i]$ all use $i$ bits from their master hash. Thus, we can implement the bins in $I_i$ to only explicitly store $\log |U| - i$ bits of each key, with the rest of the bits for the key being stored implicitly by quotienting.

A consequence of this design is that some bins are more space efficient than others. If there are $m$ bins, then the most space efficient bins use $R = \log |U| - \log m + \log s$ bits per key (ignoring space used for metadata and empty slots) and the least space efficient bins use $\log |U|$ bits per key. The fraction of bins that use $R + i$ bits per key is $\Theta(1 / 2^i)$. Thus, the total number of bits wasted by not storing exactly $R$ bits per key is
$$O\left(\sum_{i \ge 1} \frac{n}{2^i}i\right) \le O(n).$$
Critically, the fact that some keys save more bits than others only affects our space consumption by $O(n)$ bits.

\paragraph{Analyzing the total space consumption of the hash table.}  The use of quotients in place of the master hash function complicates several aspects of the analysis of Iceberg hashing. Before discussing these aspects, however, let us analyze the space consumption of the hash table, assuming the standard analysis of Iceberg hashing. Throughout the rest of the section we set $h$ to be $\Theta(\log^2 n / \log n)$, which maximizes the space efficiency of the data structure.

The number of bits used to store keys (in the front yard) of the table is
\begin{equation} nR + O(n) = \log \binom{|U|}{n} + O(n \log \log n)
\label{eq:nron}
\end{equation} bits. As shown in Theorem~\ref{thm:dynamic}, the space consumed by the backyard table is $O(n / \log n)$ bits, the space consumed by meta-data is $O(n\log \log n )$ bits, and the empty slots in front-yard bins induce at most a $1 + O(\log \log n / \log n)$ multiplicative overhead on the space needed to store the keys (that is, on \eqref{eq:nron}). Putting the pieces together, the total space consumption is
$$ \log \binom{|U|}{n} + O(n \log \log n)$$
bits.

\paragraph{Handling lack of independence in the master hash function.}  We now turn our attention to a subtle complication that arises in analyzing Iceberg hash tables that use quotienting. Because the keys are assumed to be random \emph{distinct} elements from a universe $U$, the master hashes (and thus the bin assignments) are not independent. In particular, the distinctness assumption introduces (negative) correlation between the bin assignments of keys. Since the bin assignments are no longer independent, we can no longer directly apply the Iceberg lemma.

Let $K \subset U$ be the set of keys that are ever placed into the hash table. In general, $K$ could contain all of $U$. At the cost of making each key $\Theta(\log \log n)$ bits longer, we can assume without loss of generality that $|K| \le |U| / \polylog n$ for a polylogarithmic factor of our choice. Call this the \defn{sparsity property}.

Let $y_1, \ldots, y_{|K|}$ be independently selected random elements of $U$. For the sake of analysis, we can treat the elements $x_1, \ldots, x_{|K|}$ of $K$ as being constructed via the following process: for $i = 1, \ldots, |K|$, if $y_i \not \in \{x_1, \ldots, x_{i-1} \}$, then set $x_i = y_i$ and otherwise select $x_i$ at random from $U \setminus \{x_1, \ldots, x_{i-1} \}$. Say that the key $x_i$ is \defn{dangerous} if $x_i \neq y_i$. 

When a key $x$ is inserted, say that  $x$ is \defn{vicariously dangerous} if either $x$ is dangerous or there is another key $y$ that is present and maps to the same bin as does $x$. In order to analyze the backyard of Iceberg hashing in the context of random keys (whose quotients are used as master hashes), it suffices to show that, at any given moment, the number of vicariously dangerous keys is $n / \polylog n$. All other keys can be analyzed as though the master hashes were determined by $y_1, \ldots, y_{|K|}$ (which are independent). 

\begin{lem}

Consider a moment in which there are $n$ keys in the table. Then w.s.h.p.\ in $n$, the number of dangerous keys present is $n / \polylog n$ (for a polylogarithmic factor of our choice).
\label{lem:danger}
\end{lem}

\begin{proof}

The probability that $x_i$ is dangerous is exactly $(i-1)/|U|$, and the property of being dangerous is independent between keys $x_i$. By the sparsity property, the probability $(i-1)/|U|$ is at most $1 / \polylog n$ for all keys $x_i$. The lemma therefore follows by a Chernoff bound.

\end{proof}

\begin{lem}

Consider a moment in which there are $n$ keys in the table. Then w.s.h.p.\ in $n$, the number of vicariously dangerous keys present is $n / \polylog n$ (for a polylogarithmic factor of our choice).
\label{lem:vicarious}
\end{lem}

\begin{proof}

Let $A = \{a_1, \ldots, a_n \}$ be the keys present in the table. Let $Y_1, \ldots, Y_n$ be such that $Y_i$ is the set of keys present at the time of $a_i$’s insertion. Let $Z_i = Y_i \setminus A$.

There are three ways that $a_i$ can be vicariously dangerous: 
\begin{itemize}
	\item The first case is that $a_i$ itself is dangerous. Lemma \ref{lem:danger} tells us that w.s.h.p.\ in $n$ there are at most $n / \polylog n $ dangerous keys $a_i$.
	\item The second case is that there is an element $y \in Z_i$ such that $y$ is dangerous and $a_i$ and $y$ map to the same bin. By Lemma \ref{lem:danger} (w.s.h.p./ in $n$), the number of dangerous keys in $Z_i$ is $n / \polylog n $. This, in turn, means that at most a $1 / \polylog n $ fraction of bins contain a dangerous key from $Z_i$. The probability of $a_i$ mapping to the same bin as such a key (and not itself being dangerous) is at most $1 / \polylog n$. Since the $Z_i$'s are disjoint (and ignoring the keys $a_i$ that are dangerous), these probabilities are independent across keys $a_i$. By a Chernoff bound, w.s.h.p.\ in $n$, the number of keys in this case (that are not dangerous) is $n / \polylog n$.
	\item The third case is that there is an element $y \in A \setminus \{a_i\}$ such that $y$ is dangerous and $a_i$ and $y$ map to the same bin. By Lemma \ref{lem:danger}, the number of dangerous keys in $A$ is at most $n / \polylog n$ w.s.h.p.\. Conditioning on this, each key independently has at most a $1/ \polylog n$ probability of being in this third case (and not being dangerous). By a Chernoff bound, w.s.h.p.\ in $n$, the number of keys in this case (that are not dangerous) is $n / \polylog n$.
\end{itemize}

Combining the cases completes the proof of the lemma.

\end{proof}

By Lemma \ref{lem:vicarious}, the fact that master hashes are determined by $x_1, \ldots, x_{|K|}$ (which are not independent) instead of $y_1, \ldots, y_{|K|}$ (which are independent) only affects the size of the backyard of the hash table by $n / \polylog n$ (because of 
vicariously dangerous records behaving differently in the two cases). 

\paragraph{Simulating random keys with almost random permutations.}  In order to simulate random keys, a natural approach is to apply a random permutation to the universe $U$. Constructing an efficiently describable random permutation that can be evaluated in constant time remains a significant open question. Fortunately, there do exist efficient \defn{$k$-wise $\delta$-dependent} permutations \cite{naor1999construction, luby1988construct}, that is, permutations drawn from a distribution that is $\delta$-close to being $k$-wise independent. In more detail, there exists some constant $\alpha > 0$ such that for $k = n^{\alpha}$ and $\delta = 1/ \poly n$, there is a $k$-wise $\delta$-dependent family of permutations whose members can be evaluated in constant time and described using $n^{\beta}$ bits for some $\beta < 1$. In particular, one can achieve $\delta = 1 / 2^{\Omega(\log n)}$ using Corollary 8.1 of \cite{kaplan2009derandomized} (along with the hash family of \cite{pagh2008uniform} for $f_1, f_2$),
and then, as shown by \cite{kaplan2009derandomized}, $\delta$ can be amplified to $1 / \poly n$ by
composing together $O(1)$ independently selected permutations that each satisfy 
$\delta = 1/2^{\Omega(\log n)}$.

Because $\delta = 1/ \poly n$, the fraction of the time that the hash family does not behave as $k$-wise independent can be easily absorbed into the failure probability of Iceberg hashing (assuming we are only proving a w.h.p. guarantee). On the other hand, $n^{\alpha}$-wise independence does not obviously suffice for our analysis of Iceberg hashing. Essentially the same problem was encountered previously in \cite{arbitman2010backyard}, and their solution also works here. For completeness we describe the solution below.

Let $N$ be a parameter. As in Appendix \ref{sec:hashing} (where we discuss how to construct explicit families of hash functions for Iceberg hashing), we break our table into $N^{1 - \epsilon}$ subtables for some $\epsilon$ sufficiently smaller than $\alpha$. We will guarantee that the subtables are all the same sizes as each other, up to negligible terms, which means that they can be resized synchronously with each other; this, in turn, means that each partial expansion/shrinkage can be implemented in $O(1)$ memory allocations, which allows for us to directly access all of the subtables without any extra layers of indirection. That is, the act of decomposing the hash table into $N^{1 - \epsilon}$ subtables does not hurt the cache-efficiency of our data structure.

We may assume without loss of generality that the size of the table stays in the range $[N/2,N]$ for some $N$, since every time the size of the table changes by a constant factor, we can rebuild the table (in a deamortized fashion) to accommodate the new value of $N$. Note that such a rebuild does not violate succinctness because, as we move elements from the old version of the table to the new version, the partial shrinkages that occur in the old subtables will keep them succinct until they get to small enough sizes that their space consumption is negligible. 

Keys $x$ are mapped to a subtable by performing a permutation $\pi_1(x)$ and then using the least significant $(1 - \epsilon)\log N$ bits as a subtable choice. Let $x'$ denote the most significant $\log U - \log N + \epsilon\log N$ bits of $\pi_1(x)$. Rather than storing $x$ in the subtable, it suffices to store $x'$. And rather than storing $x'$ in the subtable, we instead perform a second permutation $\pi_2(x')$ to obtain the actual key that we store in the subtable.

The second permutation $\pi_2$ can be implemented as a $N^{\alpha}$-wise $(1/ \poly n)$-dependent permutation. As in Section \ref{sec:hashing}, the small size of the subtable ensures that $N^{\alpha}$-independence suffices. The more difficult challenge is implementing $\pi_1$ so that, with high probability, each of the subtables receive $N^\epsilon + O(N^{(2 / 3)\epsilon})$ keys.

Arbitman et al. \cite{arbitman2010backyard} give an elegant solution to this problem by defining $\pi_1$ using a single-round Feistel permutation. Define the \defn{right part} $x_R$ of a key $x$ to be the least significant $(1 - \epsilon)\log N$ bits of $x$ and define the \defn{left part} $x_L$ to be the remaining bits.  Let $\mathcal{H}$ be the family of hash functions given by Pagh and Pagh \cite{pagh2008uniform} parameterized to simulate $k$-independence for $k = N / \log^2 N$ and so that each $h \in \mathcal{H}$ maps the left part $x_L$ of a key to an output of $(1 - \epsilon) \log N$ bits (which is the same number of bits in the right part $x_R$ of the key). The guarantee given by \cite{pagh2008uniform} is that for a random $h \in \mathcal{H}$, and for any given set $S$ of size $O(N / \log^2 N)$, the function $h \in \mathcal{H}$ acts fully randomly on $S$ with high probability in $N$; moreover, each hash function $h \in \mathcal{H}$ can be represented with $O(N / \log N)$ description bits and can be evaluated in constant time. Using a random $h \in \mathcal{H}$, the permutation $\pi_1(x)$ is defined by $h(x_L) \oplus x$ where $\oplus$ denotes the XOR operator. Note that $\pi_1$ changes only the least significant $(1-\epsilon)\log N$ bits of $x$, meaning that $x_L$ does not change. Thus, even though the function $h$ may not be invertible, the function $\pi_1$ is invertible (and, in fact, $\pi_1 = \pi_1^{-1}$). This ensures that $\pi_1$ is a permutation. On the other hand, as shown by Arbitman et al. \cite{arbitman2010backyard} (see their Claim 5.7), the randomness from $h$ is sufficient to ensure that $\pi_1$ distributes keys evenly among the subtables, that is, every subtable has $N^{\epsilon} + O(N^{(2/3) \epsilon})$ keys with high probability in $N$.

We remark that, since $\pi_1$ preserves $x_L$, the input to $\pi_2$ is actually just $x_L$. Thus, the subtable is selected by $h(x_L) \oplus x_R$ and then the key $\pi_2(x_L)$ is stored in the subtable.

We also remark that, although the permutations $\pi_1$ and $\pi_2$ are used to randomize the key (and thus determine the master hash), the chunk hash functions $\{g_i^{(a)}\}$ used by waterfall addressing must be generated through a separate process, and should thus be implemented using the hash-function construction given in Appendix \ref{sec:hashing}.

\paragraph{Putting the pieces together.}  To conclude the section, we give a theorem summarizing the guarantees of a quotiented Iceberg hash table.

\begin{thm}[Theorem \ref{thm:dynamic} with Quotienting] Consider a dynamic quotiented Iceberg hash table. Let $n$ be the current number of keys, and suppose $|U| \le \poly n$. Then the table consumes
$$ \log \binom{|U|}{n} + O(n \log \log n)$$
bits and supports operations which run in constant time with high probability in $n$. Additionally, the stability and cache-efficiency guarantees that hold on the non-quotiented Iceberg hash table continue to hold for the quotiented Iceberg hash table (although, of course, due to quotienting, some bits of each key may be stored implicitly based on where the key resides).
\end{thm}

We remark that the quotiented Iceberg hash table can also easily be adapted to store an $O(\log n)$-bit value for each key. If values are $j$ bits, then the table uses space 
$$ \log \binom{|U|}{n} + nj + O(n \log \log n)$$
bits.

 \section{Other Related Work on Hash Tables}
\label{sec:related}

In this section we summarize some of the milestones in past work on hash-table design. Although many of these works are also discussed earlier in the paper, we include a discussion of them all together here for completeness. 

The first hash table to achieve constant-time operations with high probability was that of Dietzfelbinger et al.~\cite{dietzfelbinger1990new} in 1990 (building on previous work by Fredman et al.~\cite{Fredman82FKS} and Dietzfelbinger et al.~\cite{dietzfelbinger88hash}). Subsequently, Pagh and Rodler \cite{Pagh:CuckooHash} introduced a much simpler hash table, namely Cuckoo hashing, that achieves constant-time queries but allows for insertions to sometimes take longer. By queuing the work to be performed in a Cuckoo hash table, and performing it incrementally, Arbitman et al.~\cite{Arbitman09Deamortized} showed how to make all operations in a Cuckoo hash table take constant time.

A separate line of work has focused on optimizing space utilization. The first dynamically-resizable, succinct (i.e., the space consumption is comparable with the information theory lower bound) hash table was proposed by Raman and Rao \cite{Raman03Succinct} in 2003, but the insertion cost was only $O(1)$ expected in the amortized sense. Demaine et al.~\cite{demaine2005dynamic} improved this to constant time in the worst case in exchange for a constant factor loss in space consumption. In 2010, Arbitman et al.~\cite{arbitman2010backyard} gave the first hash table to both be succinct and provide all worst-case costs (although it is not dynamically resizable). They used a front yard/backyard table, in which the backyard is implemented as a deamortized Cuckoo hash table, which naturally lends to a mechanism for controlling the occupancy of the backyard by moving records back to the front yard. Similar ideas were used by Bercea and Even to build hash tables for random multisets \cite{Bercea2020Filter} and for multisets \cite{Bercea2020Dictionary}. Recently, Liu et al.~\cite{liu2020succinct} presented a dictionary that, in addition to succinctness and worst-case costs, supports dynamic resizing. As in this paper, the results of \cite{arbitman2010backyard} and \cite{liu2020succinct} are presented both in terms of hash tables with high load factors and in terms of succinct data structures.

Research on external memory hashing has taken two avenues. The first is to allow for super-constant time queries in exchange for sub-constant (amortized) time inserts and deletes \cite{Iacono12Hashing,DBLP:conf/icalp/ConwayFS18, verbin2013limits}. The second is to achieve $1 + o(1)$ cache misses per operation, for both queries, insertions, and deletes~\cite{jensen2008optimality, PaghWeYi14}.  Particularly interesting is the external memory hash table by Jensen and Pagh \cite{jensen2008optimality}, which supports all operations in $1 + o(1)$ expected amortized cache misses, uses $(1 + o(1))n$ space, and is dynamically resizable, but does not achieve constant-time operations in the RAM model. Their hash table \cite{jensen2008optimality} makes use of cache-efficient resizing techniques that were previously developed by Larson~\cite{larson1980linear} for external-memory file storage (as discussed in Section \ref{sec:dynamic}, the same resizing techniques serve as a starting point in our design of waterfall addressing), which in turn extend previous work on the topic by Litwin \cite{litwin1980linear}.

A major open question is whether randomness is needed to achieve constant-time operations (see discussion in \cite{arbitman2010backyard} as well as~\cite{Sundar91,HagerupMiPa01, Ruzic08,Pagh00,patrascu2014dynamic}). In the case where the hash table is very small, the dynamic fusion tree of P\v{a}tra\c{s}cu and Thorup \cite{patrascu2014dynamic} achieves this goal, but for larger hash tables, the question remains open. This raises the simpler question of what the smallest-achievable failure probability is. Until this paper, the only known schemes to achieve subpolynomial probabilities were those of \cite{goodrich2011fully, goodrich2012cache}, resulting in a failure probability of $1 / 2 ^ {\polylog n}$. Whether these schemes are compatible with explicit families of hash functions (without amplifying the failure probability) remains an open question.

 \section*{Acknowledgments}

We would like to thank Sepehr Assadi, Rob Johnson, John Kuszmaul, Rose Silver, and Janet Vorobyeva for helpful discussions and John Owens for suggesting the name Iceberg hashing.

This research was supported in part by NSF grants CSR-1938180, CCF-2106999, CCF-2118620, CCF-2118832, CCF-2106827, CCF-1725543, CSR-1763680, CCF-1716252 and CNS-1938709, as well as an NSF GRFP fellowship and a Fannie and John Hertz Fellowship.

This research was also partially sponsored by the United States Air Force Research Laboratory and was accomplished under Cooperative Agreement Number FA8750-19-2-1000. The views and conclusions contained in this document are those of the authors and should not be interpreted as representing the official policies, either expressed or implied, of the United States Air Force or the U.S. Government. The U.S. Government is authorized to reproduce and distribute reprints for Government purposes notwithstanding any copyright notation herein.
 
\markeverypar{\the\everypar\looseness=0 }

\bibliographystyle{plain}
\bibliography{bibliography}

\appendix
\setlength{\parindent}{1.5em}

\section{An Explicit Family of Hash Functions for Iceberg Hashing} \label{sec:hashing}

In this section, we show how to implement Iceberg hashing using $O(n^{\alpha} \log n)$ random bits for a positive constant $\alpha > 0$ of our choice. As in past work \cite{arbitman2010backyard, liu2020succinct, goodrich2012cache}, the hash-function families that we use will introduce
an additional $1 / \poly n$ probability of failure, meaning that they cannot be used to offer
anything better than w.h.p.\ guarantees.

\paragraph{Reducing to the case where keys are $\Theta(\log n)$ bits.}
In general, Iceberg hashing allows for keys as large as $\Theta(\word)$ bits, where $\word$ is
the machine word size. 
We can assume without loss of generality, however, that keys are 
$\Theta(\log n)$ bits. In particular, prior to computing the hash of a key $x$, we can use
pairwise-independent hashing to map $x$ to an intermediate value $x'$ that is $\Theta(\log n)$ bits,
and then we can compute the hash of $x'$ rather than $x$. The intermediate values introduce
a $1 / \poly n$ probability of collision between pairs of keys, but 
this is easily absorbed into the failure probability of a hash table.

\paragraph{Two families of hash functions.} We will make use of two families of hash functions, both of which map a universe $U$ of size polynomial in $n$ to $\Theta(\log n)$ bits.

The first family $\mathcal{H}_1$, which is due to Pagh and Pagh \cite{pagh2008uniform} (see also related work by Dietzfelbinger and Woelfel \cite{dietzfelbinger2003almost}), offers the following guarantee for a randomly selected hash function $g \in \mathcal{H}_1$: for any fixed set $S \subset U$ of size $|S| = n^{\alpha}$, with high probability in $n$, $g$ is random on $S$. Moreover, each hash function $g \in \mathcal{H}_1$ can be represented with $O(n^{\alpha} \log n)$ description bits and can be evaluated in constant time.

The second family $\mathcal{H}_2$ uses tabulation hashing \cite{patracscu2012power}. We will use $\mathcal{H}_2$ to map records to random “buckets” in the range $[n^{1 - \varepsilon}]$ for some small constant $\varepsilon$ to be selected later. By using tabulation hashing with an appropriately small table-size parameter $c$, we can arrive at the following guarantee for a randomly selected hash function $g \in \mathcal{H}_2$: for any set $S$ of $O(n)$ records, with high probability in $n$, the number of records from $S$ that map to any given bucket is $|S| / n^{1 - \varepsilon} \pm n^{(2/3) \varepsilon}$ (see Theorem $1$ of \cite{patracscu2012power}). Moreover, each function $g \in \mathcal{H}_2$ can be represented with $O(n^{\varepsilon})$ description bits and can be evaluated in constant time.

\paragraph{Using $\mathcal{H}_1$ and $\mathcal{H}_2$ in Iceberg hashing.} Let $\alpha > 0$ be a small positive constant of our choice and let $\varepsilon > 0$ be a sufficiently small positive constant relative to $\alpha$. Let $N$ be a parameter and consider a hash table whose size stays in the range $[N^{1-\varepsilon/4}, N]$. We remark that this size restriction is without loss of generality using the window rebuild technique described in Remark \ref{rem:window-run-rebuilds} in Section \ref{sec:dynamic}.

We maintain $k = N^{1 - \varepsilon}$ Iceberg hash tables $T_1, \ldots, T_k$, each of which is managed using a single hash function $g$ drawn at random from $\mathcal{H}_1$. \footnote{We can use the hash function to generate all of the hash functions needed for an Iceberg hash table. To evaluate the $i$th hash function on a key $x$, we simply compute $g(x \circ i)$, where $\circ$ denotes concatenation.} Keys are then mapped to a random table $T_i$, where $i$ is selected using a random hash function $f$ from $\mathcal{H}_2$.

The guarantee of $\mathcal{H}_2$ ensures that all of the tables $T_1, \ldots, T_k$ have the same numbers of records assigned to them up to $\pm N^{(2/3) \varepsilon}$, which is a low order term for each table. As a consequence, we can dynamically resize all of the tables $T_1, \ldots, T_k$ in sync with one another. That is, when we perform a partial expansion or shrinkage on one of the tables, we perform it on all of them. This is important, as it eliminates the need to have $k$ pointers pointing to different data structures, and allows us to store pointers to all of our memory allocations in cache, as in Theorem \ref{thm:dynamic}.

The guarantee of $\mathcal{H}_2$, on the other hand, allows us to treat each of the tables $T_1, \ldots, T_k$ as being managed by $N^{\alpha}$-independent hash functions. Since each table $T_i$ stores at most $N^{\varepsilon}$ keys at any given moment (with high probability), and since the analysis of Iceberg hashing on $N^{\varepsilon}$ keys can be performed with $O(N^{2\varepsilon})$-wise independence (note, in particular, that the proof of the Iceberg lemma on $m$ balls requires only $O(m^2)$-wise independence so that the random variables $\alpha = \{\alpha_i\}$ and $\beta = \{beta_i\}$ are mutually independent), it follows that $N^{\alpha}$-independence suffices for the analysis of each individual table $T_i$.

Call the resulting data structure a \defn{low-randomness Iceberg hash table}. We have the following theorem.

\begin{thm}

Consider a low-randomness Iceberg hash table whose size stays in the range $[N^{1-\varepsilon/4}, N]$. Furthermore, suppose that the description bits for $g$ and $f$ fit in cache. Then the guarantees from Theorems \ref{thm:static} \ref{thm:dynamic}, \ref{thm:static_succinct}, and \ref{thm:dynamic_succinct} hold.

\end{thm}

\begin{rmk}

As was the case in Theorem \ref{thm:dynamic} (see Remark \ref{rem:window-run-rebuilds}), the size restriction on the hash table can be removed by performing random rebuilds very rarely. As in Remark \ref{rem:window-run-rebuilds} this preserves the other guarantees of the hash table.

\end{rmk}

\section{The Full Analysis of Partial Resizing}\label{app:expansion}
In this section we give the full analysis of partial expansions in an Iceberg hash table (as described in Section \ref{sec:resizing_iceberg}).

To simplify the exposition, we shall perform our analysis as though we were using waterfall addressing (rather than truncated waterfall addressing). The relevant difference is that truncated waterfall addressing is not quite uniform, selecting some bins with a $(1 + O(1 / s))$-factor greater likelihood than others. This factor is easily absorbed into the Iceberg hash table by simply reducing the entire load of the table by a factor of $1 + O(1 / s)$ (or by increasing $h$ by a factor of $1 + O(1 / s)$). Rather than carry this factor of $1 + O(1 / s)$ (on the load of the table) around with us through the analysis, we instead perform the analysis assuming uniform bin assignments, and then adjust the analysis at the end appropriately.

The next lemma shows that the guarantee from Lemma~\ref{lem:second_level}
(i.e., the analysis of the backyard in static-size Iceberg
hashing) continues to hold after a partial expansion.

\begin{lem}
  Let $h \le \polylog n$. Consider any time $t > t_2$ prior to the
  next partial expansion or contraction. Let $r_1$ be the number of bins
  prior to the partial expansion and $r_2$ be the number of bins after
  the partial expansion. Suppose that, during the partial expansion,
  the total number of records never exceeds $r_1 h$, and set
  $N = r_2 h$.

  With super high probability in $N$, at time $t$, the number of records in
  the backyard is $N / \poly(h)$. Moreover, for a
  given record $x$, the probability that $x$ hashes to a bin $g_{\text{new}}(x)$
  with a non-zero floating counter at time $t$ is at most
  $1 / \poly(h)$.
  \label{lem:dynamic_correctness1}
\end{lem}
\begin{proof}
  Consider the following two situations.
  \begin{itemize}
    \item \textbf{Situation (1):} Suppose we create an Iceberg table $T$
    consisting of $r_1$ bins with capacities $h + \tau_h$. We then
    insert into $T$ the records present at time $t_0$. Finally, we
    duplicate on table $T$ everything that happens between time $t_0$
    and $t$.
  \item \textbf{Situation (2):} Suppose we create an Iceberg table $T$
    consisting of $r_2$ bins with capacities $h + \tau_h$. We then
    insert into $T$ the records present at time $t_0$. Finally, we
    duplicate on table $T$ everything that happens between time $t_0$
    and $t$.
  \end{itemize}
  For $i \in \{1, 2\}$ and for any time $j$, let $X_{i}(j)$ denote the
  number of capacity exposers at time $j$ in Situation
  ($i$).\footnote{In both situations, time is measured from the frame
    of reference of the \emph{actual} table that the lemma is
    about. That is, in each situation, once we have inserted the
    elements present at time $t_0$, we consider that point in time to
    be $t_0$.} Let $Y_{i}(j)$ denote the set of bins $b$ at time $j$
  in Situation ($i$) such that there is at least one capacity exposer $x$ satisfying $g_{\text{new}}(x) = b$. By Lemma
  \ref{lem:capacity_floaters}, we have that
  $X_i(t_1) \le N / \poly(h)$ and $X_i(t) \le N / \poly(h)$ with super high
  probability in $N$ for both $i \in \{1, 2\}$; and that for a given
  record $x$, the probability that $g_{\text{old}}(x) \in Y_{1}(t_1)$,
  that $g_{\text{old}}(x) \in Y_1(t)$, that
  $g_{\text{new}}(x) \in Y_2(t_1)$, or that
  $g_{\text{new}}(x) \in Y_2(t)$ is at most $1 / \poly(h)$. We shall
  use these bounds as the main tools for proving our lemma.

  The process for performing partial expansions is designed so that at
  time $t$, since $t > t_2$, there are only four ways that a record $x$ can be
  in the backyard:
  \begin{enumerate}
  \item The record $x$ is grandfathered, and at the time freeze $t_1$,
    the demand counter $d$ for the bin $g_{\text{new}}(x)$ was larger than
    $g_{\text{new}}(x)$'s capacity $h + \tau_h$.
  \item The record $x$ is grandfathered, and at the time freeze $t_1$,
    there was another record $y$ such that $g_{\text{new}}(x) = g_{\text{new}}(y)$ and both
    $x$ and $y$ had the same fingerprints. 
  \item The record $x$ was inserted after time $t_1$, and when $x$ was
    inserted, there was another record $y$ such that $g_{\text{new}}(x) = g_{\text{new}}(y)$
    and both $x$ and $y$ had the same fingerprints.
  \item The record $x$ was inserted after time $t_1$, and when $x$ was
    inserted, there were no free unreserved slots in bin $g_{\text{new}}(x)$.
  \end{enumerate}

  We begin by considering the records that fall into Case (1) and for
  which $g_{\text{new}}(x)$ is not in the new chunk $C$ (call
  this Case (1a)). The basic idea in this case will be to compare our
  situation to that of Situation (1) at time $t_1$. Let $x$ denote a
  record in Case (1a), and let $Q$ denote the set of records in Case
  (1a) that reside in the backyard at time $t$. Since bin
  $b = g_{\text{new}}(x)$ is not in $C$, the demand counter $d$ for
  bin $b$ is equal to the number of records $y$ at time $t_1$ (i.e.,
  at the time freeze) such that $g_{\text{old}}(y) = b$. Moreover, the
  bin $b$ contributes only $d - (h + \tau_h)$ of those records to
  $Q$. Thus, if $X$ is the set of records at time $t_1$, then
$$|Q| = \sum_{b = 1}^{r_1} \max\left(0, |\{x \in X \mid g_{\text{old}}(x) = b\}| - (h + \tau_h)\right).$$
But this expression is a lower bound on $X_1(t_1)$, which we know is at
most $ N / \poly(h)$ with super high probability in $N$. Moreover, if we
define $P$ to be the set of bins $b$ for which
$|\{x \in X \mid g_{\text{old}}(x) = b\}| > h + \tau_h$, then we know that
$P \subseteq Y_1(t_1)$, which means that the probability of a given
record $x$ satisfying $g_{\text{old}}(x) \in P$ is $1 / \poly(h)$. This completes
the analysis of Case (1a).

Next we consider the records that fall into Case (1) and for which
$g_{\text{new}}(x)$ is in the new chunk $C$ (call this Case (1b)). The
analysis in this case is very similar to that of Case (1a), except
that we now compare to Situation (2) at time $t_1$. Let $x$ denote a
record in Case (1b), and let $Q$ denote the set of records in Case
(1b) that reside in the backyard at time $t$. Since bin
$b = g_{\text{new}}(x)$ is in $C$, the demand counter $d$ for bin $b$ is equal to
the number of records $y$ at time $t_1$ (i.e., at the time freeze)
such that $g_{\text{new}}(y) = b$. Moreover, the bin $b$ contributes only $d - (h + \tau_h)$
of those records to $Q$. Thus, if $X$ is the set of records at time $t_1$,
then
$$|Q| = \sum_{b = 1}^{r_2} \max\left( 0, |\{x \in X \mid g_{\text{new}}(x) = b\}| - (h + \tau_h)\right).$$
But this expression is also a lower bound on $X_2(t_1)$, which
we know is at most $N / \poly(h)$ with super high probability in
$N$. Moreover, if we define $P$ to be the set of bins $b$ for which
$|\{x \in X \mid g_{\text{new}}(x) = b\}| > h + \tau_h$, then we know that
$P \subseteq Y_2(t_1)$, which means that the probability of a given
record $x$ satisfying $g_{\text{new}}(x) \in P$ is $1 / \poly(h)$. This completes
the analysis of Case (1b).

The analysis of Cases (2) and (3) follows directly from Lemma
\ref{lem:fingerprints}. In particular, the number of records $x$ at
time $t$ that Cases (2) and (3) contribute to the backyard is at
most $N / \poly(h)$ with super high probability, and the probability of a
record $y$ hashing to a bin $g_{\text{new}}(y)$ containing such a record $x$ is
at most $1 / \poly(h)$.

We break Case (4) into two subcases just as we did for Case (1).
Consider the records $x$ in Case (4) and such that $g_{\text{new}}(x) \not\in C$
(call this Case (4a)). The number of such records $x$ is at most
$X_{1}(t)$, which we know is at most $N / \poly(h)$ with super high
probability. Moreover, the set $Y$ of bins containing such records $x$
satisfies $Y \subseteq Y_1(t)$, meaning that the probability of a
record $y$ hashing to a bin $g_{\text{old}}(y) \in Y$ is at most $1 / \poly(h)$.

Finally, consider the records $x$ in Case (4) and such that
$g_{\text{new}}(x) \in C$ (call this Case (4b)). The number of such records $x$
is at most $X_2(t)$, which we know is at most $N / \poly(h)$ with super high
probability. Moreover, the set $Y$ of bins containing such records $x$
satisfies $Y \subseteq Y_2(t)$, meaning that the probability of a
record $y$ hashing to a bin $g_{\text{new}}(y) \in Y$ is at most $1 / \poly(h)$.
\end{proof}

We can extend the preceding lemma to consider times
$t \in [t_1, t_2]$.

\begin{lem}
  Consider any time $t \in [t_1, t_2]$ prior to the next partial
  expansion or contraction. Let $r_1$ be the number of bins prior to the
  partial expansion and $r_2$ be the number of bins after the partial
  expansion. Suppose that, during the partial expansion, the total
  number of records never exceeds $r_1 h$, and set $N = r_2
  h$. Finally, let $k$ be the number of records $x$ in the backyard at time $t_1$ and let $p$ be the probability that a record $x$
  hashes to a bin $g_{\text{old}}(x)$ with a non-zero floating counter at time
  $t_1$.

  With super high probability in $N$, at time $t$, the number of records in
  the backyard is at most $k + N /
  \poly(h)$. Moreover, for a given record $x$, the probability that
  $x$ hashes to a bin $g_{\text{new}}(x)$ with a non-zero floating counter at
  time $t$ is at most $p + 1 / \poly(h)$.
  \label{lem:dynamic_correctness2}
\end{lem}
\begin{proof}
  This follows by the same analysis as Lemma
  \ref{lem:dynamic_correctness1}, except that we also consider a fifth
  way that records can reside in the backyard, which is that they
  resided in the backyard at time $t_1$.
\end{proof}

\paragraph{Performing partial contractions and hysteresis.}
Partial contractions can be implemented using the same approach as is
described above for partial expansions. We perform a time freeze
$t_1$ at which point we reserve space in each bin for the
records that are currently present and that wish to reside in that bin
(reserving up to $h + \tau_h$ slots in each bin). We then perform the
Reshuffling Phase in the same way as for partial expansions.

Partial expansions and contractions are performed via
hysteresis. Consider a chunk $C$ that is the $j$-th 
chunk in a doubling from $2^a$ bins to $2^{a + 1}$ bins. Let
$E = 2^a / s$ denote the number of bins in $C$. We must perform
partial expansions and contractions so that, whenever the number of
records is $h \cdot (2^a + (j - 1)E)$ or larger, the chunk $C$
is included in the table, and whenever the number of records is
$h \cdot (2^a + (j - 2)E)$ or smaller, the chunk $C$ is not
included in the table. To achieve this, whenever the number of records
reaches $h \cdot (2^a + (j - 2)E + (2/3) E)$, if $C$ is not yet
present, then we perform a partial expansion during the next $hE / 3$
operations. Likewise, whenever the number of records reaches
$h \cdot (2^a (j - 2)E + (1/3)E)$, if $C$ is present, then we perform
a partial contraction during the next $hE / 3$ operations. These
thresholds ensure that partial expansions and partial contractions do
not overlap temporally.

By combining the analyses of partial expansions and partial
contractions, we arrive at a (super) high probability guarantee in terms of the
table's current size $n$.

\begin{lem}
  Consider a dynamic Iceberg table maintained with $s = \sqrt{h}$, and
  suppose the size $n$ stays in the range such that
  $h \le \polylog n$.

  Consider a time $t$, and let $n$ be the current size of the
  table. Then w.s.h.p.\ in $n$, there are at most
  $n / \poly(h)$ records in the backyard.  Moreover, for a
  given key $x$, the probability that bin $g(x)$ has a non-zero
  floating counter is at most $1 / \poly(h)$.
  \label{lem:second_level_current_size}
\end{lem}
\begin{proof}
  This follows directly from the analysis of the probabilistic guarantees during and after each partial resize.
\end{proof}

  \paragraph{Performing cache efficient resizing.}
  Next we consider the question of how to implement partial expansions
  and partial contractions efficiently in the EM model. Suppose
  $s = \sqrt{h}$, and suppose that the size of a cache line is
  $B = \Omega(h)$. Finally, suppose that we have a cache of size at least $M = ch^{1.5}B + (s \log n)$ for some sufficiently
  large constant $c$. (Note that $s \log n$ space is
  simply for storing pointers to the chunks of memory that have been
  allocated during each partial expansion in the table's history.)

  We begin by describing how to efficiently implement the Reshuffling
  Phase of a partial expansion. Traversing the backyard and
  attempting to move grandfathered records back down to the front
  yard requires only $O(n / \poly(h))$ cache misses (w.s.h.p.), where $n$ is the current table size. Reshuffling
  records within the front yard is slightly more subtle, however,
  since we wish to move roughly $\Theta(n / s)$ records with much
  fewer than $n / s$ cache misses.

  Say that a record $x$ is \defn{promoted $k$-levels} during a partial
  expansion if, due to the partial expansion, the number of bits in $x$'s bin position that are
  determined by $m(x)$ increases by $k$.

  Partition the bins into \defn{reshuffling groups}
  where the reshuffling group of each bin is determined by the bin
  number modulo $E$.  There are $O(hs) \le O(h^{1.5})$ elements in
  each rearrangement group. Importantly, any record that is promoted
  fewer than $\log h$ levels has the property that, when it is
  promoted, its rearrangement group doesn't change (i.e., it's moved
  between two bins in the same rearrangement group). We perform the
  partial expansion group by group, loading a given reshuffling group
  into cache, and then performing the reshufflings for that
  group. Once a group is loaded into cache, moving records around
  within the group is free (in terms of cache misses). Promoting
  records more than $\log h$ levels is not free, and each such
  promotion may incur up to $O(1)$ cache misses.

  By analyzing the above scheme, we can bound the number of cache misses
  needed to perform the Reshuffling Phase.
 
  \begin{lem}
    Suppose $s = \sqrt{h}$, and suppose that the size of a cache line
    is $B = \Omega(h)$. Finally, suppose that we have a cache of size at least $M = ch^{1.5}B + \sqrt{h} \log n$ for some sufficiently large constant $c$. Let $n$ satisfying $h \le \polylog n$ be the
    current table size, and suppose
    we perform a partial expansion. Then the Reshuffling Phase can be
    implemented with $O(n / \sqrt{h})$ cache misses, w.s.h.p.\ in $n$.
    \label{lem:reshuffling}
  \end{lem}
  \begin{proof}
    The number of cache misses spent on grandfathered records in the
    backyard is $O(n / \poly(h))$ w.s.h.p.\ in $n$. The number of cache misses
    spent rearranging records within each rearrangement group is
    $O(n / h)$, since each rearrangement group is loaded into and out
    of cache once. Finally, since each record $x$ has probability at
    most $1 / h$ of being promoted $\log h$ or more levels (all at once) during the
    partial expansion\footnote{In particular, in the event that $x$ is promoted by $\log h$ or more levels, we must have that $P_a(x) = 1$ and that $P_{a - 1}(x),\ldots,P_{a - \log h + 1}(x) = 0$, where $P(x)$ is the promotion sequence and the current number of bins is in the range $(2^a,2^{a + 1}]$. This, in turn, happens with probability $1 / h$.}, we have by a Chernoff bound that the number of records $x$ that are promoted
    $\log h$ or more levels is $O(n / h)$ w.s.h.p.\ in $n$.
  \end{proof}

  So far we have described how to perform the Reshuffling Phase
  efficiently. Notice, however, that the Preprocessing Phase can be
  implemented with the same grouping approach, and that the approach
  also works for the phases of partial contractions. Thus we have the
  following lemma:
  
  \begin{lem}
    Suppose $s = \sqrt{h}$, and suppose that
    the size of a cache line is $B = \Omega(h)$. Finally, suppose that
    we have a cache of size at least $M = ch^{1.5}B + \sqrt{h} \log n$ for some
    sufficiently large constant $c$. Let $n$ satisfying $h \le \polylog n$ be the current table
    size. Then a partial expansion or contraction can be implemented to
    incur at most $O(n / h)$ cache misses w.s.h.p.\
    in $n$.
    \label{lem:resize_cache}
  \end{lem}
  
  Since each partial expansion and contraction incurs at most $O(n / h)$
  cache misses (w.s.h.p.), and is spread across
  $\Theta(hE) = n / \sqrt{h}$ operations, we can randomize on which
  operations the cache misses occur, so that each operation incurs
  only $O(1 / \sqrt{h})$ resizing cache misses in expectation. Thus we
  arrive at the following lemma:
  \begin{lem}
    Consider a dynamic Iceberg hash table whose size $n$ stays in the
    range such that $h \le O(\log n / \log \log n)$, and suppose that
    we set $s = \sqrt{h}$. Suppose that the table used
    in the backyard supports constant-time operations (w.h.p. in $n$) and has load factor at least
    $1 / \poly(h)$.  Finally, suppose that
    we have a cache of size at least $M = ch^{1.5}B + \sqrt{h} \log n$ for some
    sufficiently large constant $c$, and suppose that each bin is stored in
    a cache line of size $\Theta(B)$.

    The expected number of cache
    misses incurred by a given operation is $1 + O(1 /
    \sqrt{h})$.
    \label{lem:cache-misses-dynamic}
  \end{lem}
  \begin{proof}
    We have already shown that the expected number of cache misses
    incurred by work spent on resizing is $O(1 / \sqrt{h})$. Whenever
    a partial expansion or contraction is occurring, each record $x$ has
    an $O(1 / s) = O(1 / \sqrt{h})$ chance of having
  $g_{\text{old}}(x) \neq g_{\text{new}}(x)$, in which case an
  operation on $x$ may be forced to visit multiple bins. In the cases
  where $g_{\text{old}}(x) = g_{\text{new}}(x)$, we can analyze the
  cache misses just as in Theorem \ref{thm:static}, except that we now
  use Lemma \ref{lem:second_level_current_size} in place of Lemma
  \ref{lem:second_level}.
  \end{proof}

 We can now prove Theorem \ref{thm:dynamic}.
  
\begin{proof}[Proof of Theorem \ref{thm:dynamic}]
  The claim of cache efficiency follows from Lemma
  \ref{lem:cache-misses-dynamic}. The claim of stability follows from
  the design of the data structure.

  To prove the claim of time efficiency, we must verify that each
  partial expansion/contraction can be completed in time $O(n / s)$. 
  By Lemma
  \ref{lem:second_level_current_size}, the backyard has size
  $O(n / \poly(h)) \le O(n / s)$ with high probability, and thus we can ignore resizing time spent on records in the backyard.
  By Lemma \ref{lem:listsizes}, the time spent in the Preprocessing Phase and the Reshuffling Phase on records in the front yard 
  is $O(n / s)$ with high probability in $n$. The other time costs (not from resizing) can be analyzed
  just as in Theorem \ref{thm:static}.

  Finally we prove space efficiency. By Lemma
  \ref{lem:second_level_current_size}, the space consumed by the
  backyard is negligible. On the other hand, our scheme always
  maintains an average load of $(1 - O(1 / s))h$ on the bins in the
  table.\footnote{This differs from the static case, where we
    maintained an average load of $h$. The difference stems from (a)
    the fact that we perform resizing using hysteresis, allowing for
    the load to change by a factor of $1 \pm O(1 / s)$ before
    performing resizing, and (b) the fact that we must decrease the
    load by a factor of $1 - O(1 / s)$ in order to account for
    truncated waterfall addressing being slightly nonuniform.} Since
  each bin takes space $(1 + O(\log h / h))$, the claim of
  space efficiency follows.
  \end{proof}

\newpage 
\section{Additional Figures}\label{app:figures}

Figure \ref{fig:metas} summarizes the different types of metadata in an Iceberg hash table. In some cases, it is easiest to bound the metadata by considering the overhead per record (typically either $O(1)$ or $O(\log \log n)$ bits). In other cases, it is easier to bound the overhead on a per-bin basis (there are $n / h$ bins). Keeping these distinctions in mind, all of the bounds are straightforward to derive (and have already been derived in previous sections). The only case where a metadata takes `negative space' is the use of the quotienting technique, which saves space overall. 

Figure \ref{fig:floaters} summarizes the three ways that an element can end up in the backyard. As the data structure is stable, once an element is in the backyard, it remains there until it is next deleted (or, in some cases, until a resize occurs). 

\begin{figure}\begin{tabular}{|c | c | c| }
\hline
Metadata Type & Number of Bits Across Hash Table & Sections in Use \\ \hline
Frontyard Bins   & $O(wn + w \tau_h /h )$ & All sections \\Backyard Hash Table   & $O(wn / \poly(h))$  & All sections \\Routing Tables & $O(n \log \log n)$ & Sections 3, 4, 5, 6 \\Per-Bin Floating Counters & $O\left( nw / h \right)$ & Sections 3, 4 \\Per-Bin Vacancy Bitmaps & $O(n)$ & Sections 3, 4 \\Per-Routing-Table Floating Counters & $O(n \log \log n) $& Sections 5, 6\\Per-Routing-Table Free Lists & $O(n \log \log n)$ & Sections 5, 6\\Linked Lists for Fast Resizing & $O(n \log \log n) $& Sections 4, 5, 6\\Quotienting &$- n\log n + O(n \log \log n)$ & Section 7 (saves space)\\ \hline \end{tabular}
\caption{Space consumption of different types of metadata in hash table storing $n$ keys, assuming keys of size w$ = \Theta(\log n)$ bits, $\tau_h = \Theta(\sqrt{\log h} / h)$, and a linear-space backyard $\mathcal{T}$. For each type of metadata, we also indicate which sections make use of that metadata.}\label{fig:metas}\end{figure}

\begin{figure}
\begin{tabular}{|c|c |c|}
\hline Type & Sections & Reason \\ \hline
Capacity Floaters & Sections 3, 4, 5, 6 & key hashes to a bin with $h + \tau_h$ elements \\ 
Fingerprint Floaters & Sections 3, 4, 5, 6 & key has fingerprint colliding with another key in the same routing table \\ 
Routing Floaters & Sections 5, 6 & key hashes to a routing table (within bin) that has $> 2 \log n / \log \log n$ keys. \\ \hline
\end{tabular}
\caption{Reasons that an element can be in the backyard, along with which sections the reason applies to.}
\label{fig:floaters}
\end{figure}

\end{document}